\documentclass[english]{IEEEtran}
\usepackage[T1]{fontenc}
\usepackage[latin9]{inputenc}
\usepackage{amsmath}
\usepackage{amssymb}
\usepackage{esint}

\makeatletter

\providecommand{\tabularnewline}{\\}

\newtheorem{lemma}{Lemma}
\newtheorem{remrk}{Remark}
\newtheorem{thm}{Theorem}
\newtheorem{cor}{Corollary}

\usepackage{cite}
\makeatother

\usepackage{babel}

\begin{document}

\title{On the Asymptotic Connectivity of Random Networks under the Random
Connection Model}

\author{\begin{tabular}{cc}
Guoqiang Mao & Brian DO Anderson\tabularnewline
School of Electrical and Information Engineering & Research School of Information Sciences and Engineering\tabularnewline
The University of Sydney & Australian National University\tabularnewline
National ICT Australia & National ICT Australia\tabularnewline
Email: guoqiang@ee.usyd.edu.au & Email: brian.anderson@anu.edu.au\tabularnewline
\end{tabular}\emph{}%
\thanks{This research is funded by ARC Discovery project: DP0877562.%
}}
\maketitle
\begin{abstract}
Consider a network where all nodes are distributed on a unit square
following a Poisson distribution with known density $\rho$ and a
pair of nodes separated by an Euclidean distance $x$ are directly
connected with probability $g\left(\frac{x}{r_{\rho}}\right)$, where
$g:[0,\infty)\rightarrow[0,1]$ satisfies three conditions: rotational
invariance, non-increasing monotonicity and integral boundedness,
$r_{\rho}=\sqrt{\frac{\log\rho+b}{C\rho}}$ , $C=\int_{\Re^{2}}g\left(\left\Vert \boldsymbol{x}\right\Vert \right)d\boldsymbol{x}$
and $b$ is a constant, independent of the event that another pair
of nodes are directly connected. In this paper, we analyze the asymptotic
distribution of the number of isolated nodes in the above network
using the Chen-Stein technique and the impact of the boundary effect
on the number of isolated nodes as $\rho\rightarrow\infty$. On that
basis we derive a necessary condition for the above network to be
asymptotically almost surely connected. These results form an important
link in expanding recent results on the connectivity of the random
geometric graphs from the commonly used unit disk model to the more
generic and more practical random connection model.\end{abstract}
\thispagestyle{empty}
\pagestyle{empty}

\section{Introduction\label{sec:Introduction}}

Connectivity is one of the most fundamental properties of wireless
multi-hop networks \cite{Gupta98Critical,Xue04The,Bettstetter04On,Bettstetter02On,Hekmat06Connectivity}.
A network is said to be \emph{connected} if there is a path between
any pair of nodes. In this paper we consider the necessary condition
for an \emph{asymptotically almost surely} (a.a.s.) connected network
in $\Re^{2}$. Specifically, we investigate a network where all nodes
are distributed on a unit square $\left[-\frac{1}{2},\frac{1}{2}\right)^{2}$
following a Poisson distribution with known density $\rho$ and a
pair of nodes separated by an Euclidean distance $x$ are directly
connected with probability $g\left(\frac{x}{r_{\rho}}\right)$, independent
of the event that another pair of nodes are directly connected. Here
$g:[0,\infty)\rightarrow[0,1]$ satisfies the properties of rotational
invariance, non-increasing monotonicity and integral boundedness \cite[Chapter 6]{Franceschetti07Random,Meester96Continuum}
\footnote{Throughout this paper, we use the non-bold symbol, e.g. $x$, to denote
a scalar and the bold symbol, e.g. $\boldsymbol{x}$, to denote a
vector.%
}%
\footnote{We refer readers to \cite[Chapter 6]{Franceschetti07Random,Meester96Continuum}
for detailed discussions on the random connection model.%
}:\begin{equation}
\left\{ \begin{array}{lcl}
g\left(x\right)\leq g\left(y\right) &  & whenever\;\; x\geq y\\
0<\int_{\Re^{2}}g\left(\left\Vert \boldsymbol{x}\right\Vert \right)d\boldsymbol{x}<\infty\end{array}\right.\label{eq:conditions on g(x)}\end{equation}
where $r_{\rho}=\sqrt{\frac{\log\rho+b}{C\rho}}$ , $0<C=\int_{\Re^{2}}g\left(\left\Vert \boldsymbol{x}\right\Vert \right)d\boldsymbol{x}<\infty$,
$b$ is a constant and $\left\Vert \bullet\right\Vert $ denotes the
Euclidean norm. 

It is shown later in Section \ref{sub:The-number-of-isolated-nodes-torus}
that the conditions in \eqref{eq:conditions on g(x)} imply $g\left(x\right)=o_{x}\left(\frac{1}{x^{2}}\right)$
where the symbol $o_{x}$ is defined shortly later. In this paper
we further require $g$ to satisfy a slightly more restrictive condition
that $g\left(x\right)=o_{x}\left(\frac{1}{x^{2}\log^{2}x}\right)$
and the implications of such more restrictive condition become clear
in the analysis of Section \ref{sub:The-number-of-isolated-nodes-torus},
particularly in Remark \ref{rem:the need for tighter condition on g}.
The condition $g\left(x\right)=o_{x}\left(\frac{1}{x^{2}\log^{2}x}\right)$
is only slightly more restrictive than the condition $g\left(x\right)=o_{x}\left(\frac{1}{x^{2}}\right)$
in that for an arbitrarily small positive constant $\varepsilon$,
$\frac{1}{x^{2+\varepsilon}}=o_{x}\left(\frac{1}{x^{2}\log^{2}x}\right)$.

The reason for choosing this particular form of $r_{\rho}$ is that
the analysis becomes nontrivial when $b$ is a constant. Other forms
of $r_{\rho}$ can be accommodated by allowing $b\rightarrow\infty$
or $b\rightarrow-\infty$, e.g. $b$ becomes a function of $\rho$,
as $\rho\rightarrow\infty$. We discuss these situations separately
in Section \ref{sec:The-Necessary-Condition}. 

Denote the above network by $\mathcal{G}\left(\mathcal{X}_{\rho},g_{\rho}\right)$.
It is obvious that under a \emph{unit disk model} where $g(x)=1$
for $x\leq1$ and $g(x)=0$ for $x>1$, $r_{\rho}$ corresponds to
the transmission range for connectivity \cite{Gupta98Critical}. Thus
the above model easily incorporates the unit disk model as a special
case. A similar conclusion can also be drawn for the log-normal connection
model.

The following notations and definitions are used:
\begin{itemize}
\item $f\left(z\right)=o_{z}\left(h\left(z\right)\right)$ iff (if and only
if) $\lim_{z\rightarrow\infty}\frac{f\left(z\right)}{h\left(z\right)}=0$;
\item $f\left(z\right)\sim_{z}h\left(z\right)$ iff $\lim_{z\rightarrow\infty}\frac{f\left(z\right)}{h\left(z\right)}=1$;
\item An event $\xi_{z}$ depending on $z$ is said to occur a.a.s. if its
probability tends to one as $z\rightarrow\infty$.
\end{itemize}
The above definition applies whether the argument $z$ is continuous
or discrete, e.g. assuming integer values.

The contributions of this paper are: firstly using the Chen-Stein
technique \cite{Arratia90Poisson,Barbour03Poisson}, we show that
the distribution of the number of isolated nodes in $\mathcal{G}\left(\mathcal{X}_{\rho},g_{\rho}\right)$
asymptotically converges to a Poisson distribution with mean $e^{-b}$
as $\rho\rightarrow\infty$; secondly we show that the number of isolated
nodes due to the boundary effect in $\mathcal{G}\left(\mathcal{X}_{\rho},g_{\rho}\right)$
is \emph{a.a.s.} zero, i.e. the boundary effect has asymptotically
vanishing impact on the number of isolated nodes; finally we derive
the necessary condition for $\mathcal{G}\left(\mathcal{X}_{\rho},g_{\rho}\right)$
to be a.a.s. connected as $\rho\rightarrow\infty$ under a generic
connection model, which includes the widely used unit disk model and
log-normal connection model as its two special examples. 

The rest of the paper is organized as follows: Section \ref{sec:Isolated nodes torus}
analyzes the distribution of the number of isolated nodes on a torus;
Section \ref{sec:The-Impact-of-boundary effect} evaluates the impact
of the boundary effect on the number of isolated nodes; Section \ref{sec:The-Necessary-Condition}
provides the necessary condition for $\mathcal{G}\left(\mathcal{X}_{\rho},g_{\rho}\right)$
to be a.a.s connected; Section \ref{sec:Related-Work} reviews related
work in the area. Discussions on the results and future work suggestions
appear Section \ref{sec:Conclusions-and-Further}.

\section{The Distribution of The Number Of Isolated Nodes on A Torus \label{sec:Isolated nodes torus}}

Denote by $\mathcal{G}^{T}\left(\mathcal{X}_{\rho},g_{\rho}\right)$
a network with the same node distribution and connection model as
$\mathcal{G}\left(\mathcal{X}_{\rho},g_{\rho}\right)$ except that
nodes in $\mathcal{G}^{T}\left(\mathcal{X}_{\rho},g_{\rho}\right)$
are distributed on a unit torus $\left[-\frac{1}{2},\frac{1}{2}\right)^{2}$.
In this section, we analyze the distribution of the number of isolated
nodes in $\mathcal{G}^{T}\left(\mathcal{X}_{\rho},g_{\rho}\right)$.
With minor abuse of the terminology, we use $A^{T}$ to denote both
the unit torus itself and the area of the unit torus, and in the latter
case, $A^{T}=1$.

\subsection{Difference between a torus and a square}

The unit torus $\left[-\frac{1}{2},\frac{1}{2}\right)^{2}$ that is
commonly used in random geometric graph theory is essentially the
same as a unit square $\left[-\frac{1}{2},\frac{1}{2}\right)^{2}$
except that the distance between two points on a torus is defined
by their\emph{ toroidal distance}, instead of Euclidean distance.
Thus a pair of nodes in $\mathcal{G}^{T}\left(\mathcal{X}_{\rho},g_{\rho}\right)$,
located at $\boldsymbol{x}_{1}$ and $\boldsymbol{x}_{2}$ respectively,
are directly connected with probability $g\left(\frac{\left\Vert \boldsymbol{x}_{1}-\boldsymbol{x}_{2}\right\Vert ^{T}}{r_{\rho}}\right)$
where $\left\Vert \boldsymbol{x}_{1}-\boldsymbol{x}_{2}\right\Vert ^{T}$
denotes the \emph{toroidal distance} between the two nodes. For a
unit torus $A^{T}=\left[-\frac{1}{2},\frac{1}{2}\right)^{2}$, the
toroidal distance is given by \cite[p. 13]{Penrose03Random}:\begin{equation}
\left\Vert \boldsymbol{x}_{1}-\boldsymbol{x}_{2}\right\Vert ^{T}\triangleq\min\left\{ \left\Vert \boldsymbol{x}_{1}+\boldsymbol{z}-\boldsymbol{x}_{2}\right\Vert :\boldsymbol{z}\in\mathbb{Z}^{2}\right\} \label{eq:definition of toroidal distance in a unit torus}\end{equation}
The toroidal distance between points on a torus of any other size
can be computed analogously. Such treatment allows nodes located near
the boundary to have the same number of connections \emph{probabilistically}
as a node located near the center. Therefore it allows the removal
of the boundary effect that is present in a square. \emph{The consideration
of a torus implies that there is no need to consider special cases
occurring near the boundary of the region and that events inside the
region do not depend on the particular location inside the region.}
This often simplifies the analysis however. From now on, we use the
same symbol, $A$, to denote a torus and a square. Whenever the difference
between a torus and a square affects the parameter being discussed,
we use superscript $^{T}$ (respectively $^{S}$) to mark the parameter
in a torus (respectively a square). 

We note the following relation between toroidal distance and Euclidean
distance on a square area centered at the origin:\begin{eqnarray}
\left\Vert \boldsymbol{x}_{1}-\boldsymbol{x}_{2}\right\Vert ^{T}\leq\left\Vert \boldsymbol{x}_{1}-\boldsymbol{x}_{2}\right\Vert \;\;\textrm{and}\;\;\left\Vert \boldsymbol{x}\right\Vert ^{T}=\left\Vert \boldsymbol{x}\right\Vert \label{eq:property of toroidal distance 1}\end{eqnarray}
which will be used in the later analysis.

\subsection{Properties of isolated nodes on a torus\label{sub:The-number-of-isolated-nodes-torus}}

Divide the unit torus into $m^{2}$ non-overlapping squares each with
size $\frac{1}{m^{2}}$. Denote the $i_{m}^{th}$ square by $A_{i_{m}}$.
Define two sets of indicator random variables $J_{i_{m}}^{T}$ and
$I_{i_{m}}^{T}$ with $i_{m}\in\Gamma_{m}\triangleq\{1,\ldots m^{2}\}$,
where $J_{i_{m}}^{T}=1$ iff there exists exactly one node in $A_{i_{m}}$,
otherwise $J_{i_{m}}^{T}=0$; $I_{i_{m}}^{T}=1$ iff there is exactly
one node in $A_{i_{m}}$ \emph{and} that node is isolated, $I_{i_{m}}^{T}=0$
otherwise. Obviously $J_{i_{m}}^{T}$ is independent of $J_{j_{m}}^{T},j_{m}\in\Gamma_{m}\backslash\left\{ i_{m}\right\} $.
Denote the center of $A_{i_{m}}^{T}$ by $\boldsymbol{x}_{i_{m}}$
and without loss of generality we assume that when $J_{i_{m}}^{T}=1$,
the associated node in $A_{i_{m}}$ is at $\boldsymbol{x}_{i_{m}}$%
\footnote{In this paper we are mainly concerned with the case that $m\rightarrow\infty$,
i.e. the size of the squer is varnishingly small. Therefore the actual
position of the node in the square is not important.%
}. Observe that for any fixed $m$, the values of $Pr\left(I_{i_{m}}^{T}=1\right)$
and $Pr\left(J_{i_{m}}^{T}=1\right)$ do not depend on the particular
index $i_{m}$ on a torus. However both the set of indices $\Gamma_{m}$
and a particular index $i_{m}$ depend on $m$. As $m$ changes, the
square associated with $I_{i_{m}}^{T}$ and $J_{i_{m}}^{T}$ also
changes. Without causing ambiguity, we drop the explicit dependence
on $m$ in our notations for convenience. As an easy consequence of
the Poisson node distribution, \begin{equation}
\lim_{m\rightarrow\infty}\frac{Pr\left(J_{i}^{T}=1\right)}{\frac{\rho}{m^{2}}}=1\label{eq:Prob of having a node in S_i}\end{equation}
and as $m\rightarrow\infty$, the probability that there is more than
one node in $A_{i}$ becomes vanishingly small compared to $Pr\left(J_{i}^{T}=1\right)$.
Further, using the relationship that\begin{equation}
Pr\left(I_{i}^{T}=1\right)=Pr\left(I_{i}^{T}=1|J_{i}^{T}=1\right)Pr\left(J_{i}^{T}=1\right)\label{eq:Prob Ii1 Bayes}\end{equation}
it can be shown that\begin{eqnarray}
 &  & Pr\left(I_{i}^{T}=1\right)\nonumber \\
 & = & Pr\left(J_{i}^{T}=1\right)\nonumber \\
 & \times & \prod_{j\in\Gamma\backslash\left\{ i\right\} }\left[Pr\left(J_{j}^{T}=1\right)\left(1-g\left(\frac{\left\Vert \boldsymbol{x}_{i}-\boldsymbol{x}_{j}\right\Vert ^{T}}{r_{\rho}}\right)\right)\right.\nonumber \\
 & + & \left(1-Pr\left(J_{j}^{T}=1\right)-o_{m}\left(Pr\left(J_{j}^{T}=1\right)\right)\right)\nonumber \\
 & + & \left.o_{m}\left(Pr\left(J_{j}^{T}=1\right)\left(1-g\left(\frac{\left\Vert \boldsymbol{x}_{i}-\boldsymbol{x}_{j}\right\Vert ^{T}}{r_{\rho}}\right)\right)\right)\right]\label{eq:Prob Isolated Node Discrete}\end{eqnarray}
In \eqref{eq:Prob Isolated Node Discrete}, the term $Pr\left(J_{j}^{T}=1\right)\left(1-g\left(\frac{\left\Vert \boldsymbol{x}_{i}-\boldsymbol{x}_{j}\right\Vert ^{T}}{r_{\rho}}\right)\right)$
represents the probability of the event that there is a node in $A_{j}$
\emph{and }that node is not directly connected to the node in $A_{i}$,
the term $\left(1-Pr\left(J_{j}^{T}=1\right)-o_{m}\left(Pr\left(J_{j}^{T}=1\right)\right)\right)$
represents the probability of the event that there is no node in $A_{j}$
and the last term accounts for the situation that there is more than
one node in $A_{j}$. It then follows that \begin{eqnarray}
 &  & \lim_{m\rightarrow\infty}Pr\left(I_{i}^{T}=1|J_{i}=1\right)\nonumber \\
 & = & \lim_{m\rightarrow\infty}\prod_{j\in\Gamma\backslash\left\{ i\right\} }\left[1-Pr\left(J_{j}^{T}=1\right)g\left(\frac{\left\Vert \boldsymbol{x}_{i}-\boldsymbol{x}_{j}\right\Vert ^{T}}{r_{\rho}}\right)\right]\nonumber \\
 & = & e^{-\int_{A}\rho g\left(\frac{\left\Vert \boldsymbol{x}-\boldsymbol{x}_{i}\right\Vert ^{T}}{r_{\rho}}\right)d\boldsymbol{x}}\label{eq:con prob isolated node}\\
 & = & e^{-\int_{A}\rho g\left(\frac{\left\Vert \boldsymbol{x}\right\Vert ^{T}}{r_{\rho}}\right)d\boldsymbol{x}}\label{eq:con prob isolated translation}\end{eqnarray}
where \eqref{eq:con prob isolated translation} results from \eqref{eq:con prob isolated node}
due to nodes being distributed on a torus. Further, using \eqref{eq:Prob of having a node in S_i},
\eqref{eq:Prob Ii1 Bayes} and \eqref{eq:con prob isolated translation},
it is evident that\begin{equation}
Pr\left(I_{i}^{T}=1\right)\sim_{m}\frac{\rho}{m^{2}}e^{-\int_{A}\rho g\left(\frac{\left\Vert \boldsymbol{x}\right\Vert ^{T}}{r_{\rho}}\right)d\boldsymbol{x}}\label{eq:prob isolated node}\end{equation}

Define $W_{m}^{T}=\sum_{i=1}^{m^{2}}I_{i}^{T}$ and $W^{T}=\lim_{m\rightarrow\infty}W_{m}^{T}$,
where $W^{T}$ has the meaning of the total number of isolated nodes
in $A$. It then follows that\begin{equation}
E\left(W^{T}\right)=\lim_{m\rightarrow\infty}E\left(W_{m}^{T}\right)=\rho e^{-\int_{A}\rho g\left(\frac{\left\Vert \boldsymbol{x}\right\Vert ^{T}}{r_{\rho}}\right)d\boldsymbol{x}}\label{eq:expected number of isolated nodes finite network}\end{equation}

It can be shown that

\begin{eqnarray}
 &  & \lim_{\rho\rightarrow\infty}\rho e^{-\int_{D\left(\boldsymbol{0},r_{\rho}^{1-\varepsilon}\right)}\rho g\left(\frac{\left\Vert \boldsymbol{x}\right\Vert }{r_{\rho}}\right)d\boldsymbol{x}}\nonumber \\
 & = & \lim_{\rho\rightarrow\infty}\rho e^{-\rho r_{\rho}^{2}\int_{D\left(\boldsymbol{0},r_{\rho}^{-\varepsilon}\right)}g\left(\left\Vert \boldsymbol{x}\right\Vert \right)d\boldsymbol{x}}\nonumber \\
 & = & \lim_{\rho\rightarrow\infty}\rho e^{-\rho r_{\rho}^{2}\left(C-\int_{\Re^{2}\backslash D\left(\boldsymbol{0},r_{\rho}^{-\varepsilon}\right)}g\left(\left\Vert \boldsymbol{x}\right\Vert \right)d\boldsymbol{x}\right)}\nonumber \\
 & = & e^{-b}\lim_{\rho\rightarrow\infty}e^{\frac{\log\rho+b}{C}\int_{r_{\rho}^{-\varepsilon}}^{\infty}2\pi xg\left(x\right)dx}=e^{-b}\label{eq:an analysis of the trunction effect}\end{eqnarray}
where $D\left(\boldsymbol{0},x\right)$ denotes a disk centered at
the origin and with a radius $x$, $\varepsilon$ is a small positive
constant, and the last step results because \begin{eqnarray}
 &  & \lim_{\rho\rightarrow\infty}\frac{\int_{r_{\rho}^{-\varepsilon}}^{\infty}2\pi xg\left(x\right)dx}{\frac{1}{\log\rho+b}}\nonumber \\
 & = & \lim_{\rho\rightarrow\infty}\frac{\pi\varepsilon r_{\rho}^{-\varepsilon}g\left(r_{\rho}^{-\varepsilon}\right)r_{\rho}^{-\varepsilon-2}\frac{\log\rho+b-1}{C\rho^{2}}}{\frac{1}{\rho\left(\log\rho+b\right)^{2}}}\label{eq:little's rule trunction}\\
 & = & \lim_{\rho\rightarrow\infty}\pi\varepsilon\left(\log\rho+b\right)^{2}r_{\rho}^{-2\varepsilon}o_{\rho}\left(\frac{1}{r_{\rho}^{-2\varepsilon}\log^{2}\left(r_{\rho}^{-2\varepsilon}\right)}\right)=0\nonumber \end{eqnarray}
where L'H�pital's rule is used in reaching \eqref{eq:little's rule trunction}
and in the third step $g\left(x\right)=o_{x}\left(\frac{1}{x^{2}\log^{2}x}\right)$
is used. As a consequence of \eqref{eq:property of toroidal distance 1},
\eqref{eq:expected number of isolated nodes finite network}, \eqref{eq:an analysis of the trunction effect}
and that $e^{-b}=\lim_{\rho\rightarrow\infty}\rho e^{-\int_{\Re^{2}}\rho g\left(\frac{\left\Vert \boldsymbol{x}\right\Vert }{r_{\rho}}\right)d\boldsymbol{x}}\leq\lim_{\rho\rightarrow\infty}\rho e^{-\int_{A}\rho g\left(\frac{\left\Vert \boldsymbol{x}\right\Vert }{r_{\rho}}\right)d\boldsymbol{x}}\leq\lim_{\rho\rightarrow\infty}\rho e^{-\int_{D\left(\boldsymbol{0},r_{\rho}^{1-\varepsilon}\right)}\rho g\left(\frac{\left\Vert \boldsymbol{x}\right\Vert }{r_{\rho}}\right)d\boldsymbol{x}}=e^{-b}$:

\begin{eqnarray}
\lim_{\rho\rightarrow\infty}E\left(W^{T}\right) & = & e^{-b}\label{eq:Expected number of isolated nodes asymptotic}\end{eqnarray}
The above analysis is summarized Lemma \ref{lem:Expected Isolated nodes torus}.
\begin{lemma}
\label{lem:Expected Isolated nodes torus}The expected number of isolated
nodes in $\mathcal{G}^{T}\left(\mathcal{X}_{\rho},g_{\rho}\right)$
is $\rho e^{-\int_{A}\rho g\left(\frac{\left\Vert \boldsymbol{x}\right\Vert ^{T}}{r_{\rho}}\right)d\boldsymbol{x}}$.
As $\rho\rightarrow\infty$, the expected number of isolated nodes
in $\mathcal{G}^{T}\left(\mathcal{X}_{\rho},g_{\rho}\right)$ converges
to $e^{-b}$.\end{lemma}
\begin{remrk}
\label{rem:the need for tighter condition on g}Using \eqref{eq:conditions on g(x)},
it can be shown that $C=\int_{\Re^{2}}g\left(\left\Vert \boldsymbol{x}\right\Vert \right)d\boldsymbol{x}\geq\lim_{z\rightarrow\infty}\int_{0}^{z}2\pi xg\left(z\right)dx=\lim_{z\rightarrow\infty}\pi z^{2}g\left(z\right)$.
Therefore $\lim_{z\rightarrow\infty}\frac{g\left(z\right)}{\frac{1}{z^{2}}}\leq\frac{C}{\pi}$.
It can then be shown that the only possibility is $\lim_{z\rightarrow\infty}\frac{g\left(z\right)}{\frac{1}{z^{2}}}=0$
and that the other possibilities where $\lim_{z\rightarrow\infty}\frac{g\left(z\right)}{\frac{1}{z^{2}}}\neq0$
can be ruled out by contradiction with \eqref{eq:conditions on g(x)}.
Thus \begin{equation}
g\left(x\right)=o_{x}\left(1/x^{2}\right)\label{eq:scaling property of g(x)}\end{equation}
Further the condition $g\left(x\right)=o_{x}\left(\frac{1}{x^{2}\log^{2}x}\right)$
is only required for $\rho r_{\rho}^{2}\int_{\Re^{2}\backslash D\left(\boldsymbol{0},r_{\rho}^{-\varepsilon}\right)}g\left(\left\Vert \boldsymbol{x}\right\Vert \right)d\boldsymbol{x}$
to asymptotically converge to $0$, where the term $\rho r_{\rho}^{2}\int_{\Re^{2}\backslash D\left(\boldsymbol{0},r_{\rho}^{-\varepsilon}\right)}g\left(\left\Vert \boldsymbol{x}\right\Vert \right)d\boldsymbol{x}$
is associated with (the removal of) connections between a node at
$\boldsymbol{0}$ and other nodes outside $D\left(\boldsymbol{0},r_{\rho}^{-\varepsilon}\right)$.
Evaluation of $\rho r_{\rho}^{2}\int_{\Re^{2}\backslash D\left(\boldsymbol{0},r_{\rho}^{-\varepsilon}\right)}g\left(\left\Vert \boldsymbol{x}\right\Vert \right)d\boldsymbol{x}$
for an area larger than $D\left(\boldsymbol{0},r_{\rho}^{-\varepsilon}\right)$
(but not greater than $A_{\rho}$) does not remove the need for the
condition. Thus the more restrictive requirement on $g$ that $g\left(x\right)=o_{x}\left(\frac{1}{x^{2}\log^{2}x}\right)$
is attributable to the \emph{truncation effect }that arises when considering
connectivity in a (asymptotically infinite) finite region instead
of an infinite area. 
\end{remrk}
Now consider the event $I_{i}^{T}I_{j}^{T}=1,i\neq j$ conditioned
on the event that $J_{i}^{T}J_{j}^{T}=1$, meaning that both nodes
having been placed inside $A_{i}$ and $A_{j}$ respectively are isolated.
Following the same steps leading to \eqref{eq:con prob isolated translation},
it can be shown that \begin{eqnarray}
 &  & \lim_{m\rightarrow\infty}Pr\left(I_{i}^{T}I_{j}^{T}=1|J_{i}^{T}J_{j}^{T}=1\right)\nonumber \\
 & = & \left(1-g\left(\frac{\left\Vert \boldsymbol{x}_{i}-\boldsymbol{x}_{j}\right\Vert ^{T}}{r_{\rho}}\right)\right)\nonumber \\
 & \times & \exp\left[-\int_{A}\rho\left(g\left(\frac{\left\Vert \boldsymbol{x}-\boldsymbol{x}_{i}\right\Vert ^{T}}{r_{\rho}}\right)+g\left(\frac{\left\Vert \boldsymbol{x}-\boldsymbol{x}_{j}\right\Vert ^{T}}{r_{\rho}}\right)\right.\right.\nonumber \\
 & - & \left.\left.g\left(\frac{\left\Vert \boldsymbol{x}-\boldsymbol{x}_{i}\right\Vert ^{T}}{r_{\rho}}\right)g\left(\frac{\left\Vert \boldsymbol{x}-\boldsymbol{x}_{j}\right\Vert ^{T}}{r_{\rho}}\right)\right)d\boldsymbol{x}\right]\label{eq:joint distribution of isolated node events}\end{eqnarray}
where the term $\left(1-g\left(\frac{\left\Vert \boldsymbol{x}_{i}-\boldsymbol{x}_{j}\right\Vert ^{T}}{r_{\rho}}\right)\right)$
is due to the consideration that the two nodes located inside $A_{i}$
and $A_{j}$ cannot be directly connected in order for both nodes
to be isolated. Observe also that: \begin{eqnarray}
 &  & Pr\left(I_{i}^{T}I_{j}^{T}=1\right)\nonumber \\
 & = & Pr\left(J_{i}^{T}J_{j}^{T}=1\right)Pr\left(I_{j}^{T}I_{j}^{T}=1|J_{i}^{T}J_{j}^{T}=1\right)\label{eq:Joint distribution of I_i and I_j}\end{eqnarray}
Now using \eqref{eq:Prob of having a node in S_i}, \eqref{eq:prob isolated node},
\eqref{eq:joint distribution of isolated node events} and \eqref{eq:Joint distribution of I_i and I_j},
it can be established that\begin{eqnarray}
 &  & Pr\left(I_{i}^{T}=1|I_{j}^{T}=1\right)\nonumber \\
 & \sim_{m} & \frac{\rho}{m^{2}}\left(1-g\left(\frac{\left\Vert \boldsymbol{x}_{i}-\boldsymbol{x}_{j}\right\Vert ^{T}}{r_{\rho}}\right)\right)\nonumber \\
 & \times & e^{-\int_{A}\rho\left(g\left(\frac{\left\Vert \boldsymbol{x}-\boldsymbol{x}_{i}\right\Vert ^{T}}{r_{\rho}}\right)-g\left(\frac{\left\Vert \boldsymbol{x}-\boldsymbol{x}_{i}\right\Vert ^{T}}{r_{\rho}}\right)g\left(\frac{\left\Vert \boldsymbol{x}-\boldsymbol{x}_{j}\right\Vert ^{T}}{r_{\rho}}\right)\right)d\boldsymbol{x}}\label{eq:conditional distribution of I_i|I_j}\end{eqnarray}
\begin{eqnarray}
 &  & \lim_{m\rightarrow\infty}\frac{Pr\left(I_{i}^{T}I_{j}^{T}=1\right)}{Pr\left(I_{i}^{T}=1\right)Pr\left(I_{j}^{T}=1\right)}\nonumber \\
 & = & (1-g(\frac{\left\Vert \boldsymbol{x}_{i}-\boldsymbol{x}_{j}\right\Vert ^{T}}{r_{\rho}}))e^{\int_{A}\rho g(\frac{\left\Vert \boldsymbol{x}-\boldsymbol{x}_{i}\right\Vert ^{T}}{r_{\rho}})g(\frac{\left\Vert \boldsymbol{x}-\boldsymbol{x}_{j}\right\Vert ^{T}}{r_{\rho}})d\boldsymbol{x}}\label{eq:ratio correlation}\end{eqnarray}
Using \eqref{eq:Prob of having a node in S_i}, \eqref{eq:con prob isolated node},
\eqref{eq:joint distribution of isolated node events}, \eqref{eq:Joint distribution of I_i and I_j}
and the above equation, it can also be obtained that

\begin{eqnarray}
 &  & \lim_{m\rightarrow\infty}\frac{Pr\left(I_{i}^{T}=1,I_{j}^{T}=0\right)}{Pr\left(I_{i}^{T}=1\right)Pr\left(I_{j}^{T}=0\right)}\nonumber \\
 & = & \lim_{m\rightarrow\infty}\frac{Pr\left(I_{i}^{T}=1\right)-Pr\left(I_{i}^{T}I_{j}^{T}=1\right)}{Pr\left(I_{i}^{T}=1\right)Pr\left(I_{j}^{T}=0\right)}\nonumber \\
 & = & \lim_{m\rightarrow\infty}\left(1-\frac{\rho}{m^{2}}e^{-\int_{A}\rho g\left(\frac{\left\Vert \boldsymbol{x}-\boldsymbol{x}_{j}\right\Vert ^{T}}{r_{\rho}}\right)d\boldsymbol{x}}\right)^{-1}\nonumber \\
 & \times & \left[1-\frac{\rho}{m^{2}}\left(1-g\left(\frac{\left\Vert \boldsymbol{x}_{i}-\boldsymbol{x}_{j}\right\Vert ^{T}}{r_{\rho}}\right)\right)\right.\nonumber \\
 & \times & \left.e^{-\int_{A}\rho g\left(\frac{\left\Vert \boldsymbol{x}-\boldsymbol{x}_{j}\right\Vert ^{T}}{r_{\rho}}\right)\left(1-g\left(\frac{\left\Vert \boldsymbol{x}-\boldsymbol{x}_{i}\right\Vert ^{T}}{r_{\rho}}\right)\right)d\boldsymbol{x}}\right]\label{eq:relation isolated node and non-isolated}\end{eqnarray}

\subsection{The asymptotic distribution of the number of isolated nodes\label{sec:The-Distribution-of-the-number-of-isolated-nodes} }

On the basis of the discussion in the last subsection, in this subsection
we consider the distribution of the number of isolated nodes in $\mathcal{G}^{T}\left(\mathcal{X}_{\rho},g_{\rho}\right)$
as $\rho\rightarrow\infty$ . Our analysis relies on the use of the
Chen-Stein bound \cite{Arratia90Poisson,Barbour03Poisson}. The Chen-Stein
bound is named after the work of Stein \cite{Stein72A} and Chen \cite{Chen75An,Chen75Poisson}.
It is well known that the number of occurrences of \emph{independently}
distributed \emph{rare} events in a region can often be accurately
approximated by a Poisson distribution \cite{Barbour03Poisson}. In
\cite{Stein72A}, Stein developed a novel method for showing the convergence
in distribution to the normal of the sum of a number of \emph{dependent}
random variables. In \cite{Chen75An,Chen75Poisson} Chen applied Stein's
ideas in the Poisson setting and derived an upper bound on the \emph{total
variation distance}, a concept defined in the theorem statement below,
between the distribution of the sum of a number of \emph{dependent}
random indicator variables and the associated Poisson distribution.
The following theorem gives a formal statement of the Chen-Stein bound.
\begin{thm}
\label{thm:Chen-Stein Bound}\cite[Theorem 1.A ]{Barbour03Poisson}
For a set of indicator random variables $I_{i},\; i\in\Gamma$, define
$W\triangleq\sum_{i\in\Gamma}I_{i}$, $p_{i}\triangleq E\left(I_{i}\right)$
and $\lambda\triangleq E\left(W\right)$. For any choice of the index
set $\Gamma_{s,i}\subset\Gamma$, $\Gamma_{s,i}\cap\{i\}=\{\textrm{�}\}$,\begin{eqnarray*}
 &  & d_{TV}\left(\mathcal{L}\left(W\right),Po\left(\lambda\right)\right)\\
 & \leq & \sum_{i\in\Gamma}\left[\left(p_{i}^{2}+p_{i}E\left(\sum_{j\in\Gamma_{s,i}}I_{j}\right)\right)\right]min\left(1,\frac{1}{\lambda}\right)\\
 & + & \sum_{i\in\Gamma}E\left(I_{i}\sum_{j\in\Gamma_{s,i}}I_{j}\right)min\left(1,\frac{1}{\lambda}\right)\\
 & + & \sum_{i\in\Gamma}E\left|E\left\{ I_{i}\left|\left(I_{j},j\in\Gamma_{w,i}\right)\right.\right\} -p_{i}\right|min\left(1,\frac{1}{\sqrt{\lambda}}\right)\end{eqnarray*}
where $\mathcal{L}\left(W\right)$ denotes the distribution of $W$,
$Po\left(\lambda\right)$ denotes a Poisson distribution with mean
$\lambda$, $\Gamma_{w,i}=\Gamma\backslash\left\{ \Gamma_{s,i}\cup\{i\}\right\} $
and $d_{TV}$ denotes the total variation distance. The total variation
distance between two probability distributions $\alpha$ and $\beta$
on $\mathbb{Z}^{+}$ is defined by \[
d_{TV}\left(\alpha,\beta\right)\triangleq\sup\left\{ \left|\alpha\left(A\right)-\beta\left(A\right)\right|:A\subset\mathbb{Z}^{+}\right\} \]

\end{thm}
For convenience, we separate the bound in Theorem \ref{thm:Chen-Stein Bound}
into three terms $b_{1}min\left(1,\frac{1}{\lambda}\right)$, $b_{2}min\left(1,\frac{1}{\lambda}\right)$
and $b_{3}min\left(1,\frac{1}{\sqrt{\lambda}}\right)$ where $b_{1}\triangleq\sum_{i\in\Gamma}\left[\left(p_{i}^{2}+p_{i}E\left(\sum_{j\in\Gamma_{s,i}}I_{j}\right)\right)\right]$,
$b_{2}\triangleq\sum_{i\in\Gamma}E\left(I_{i}\sum_{j\in\Gamma_{s,i}}I_{j}\right)$
and $b_{3}\triangleq\sum_{i\in\Gamma}E\left|E\left\{ I_{i}\left|\left(I_{j},j\in\Gamma_{w,i}\right)\right.\right\} -p_{i}\right|$.

The set of indices $\Gamma_{s,i}$ is often chosen to contain all
those $j$, other than $i$, for which $I_{j}$ is {}``strongly''
dependent on $I_{i}$ and the set $\Gamma_{w,i}$ often contains all
other indices apart from $i$ for which $I_{j},j\in\Gamma_{W,i}$
are at most {}``weakly'' dependent on $I_{i}$ \cite{Arratia90Poisson}.
In many applications, by a suitable choice of $\Gamma_{s,i}$ the
$b_{3}$ term can be easily made to be $0$ and the evaluation of
the $b_{1}$ and $b_{2}$ terms involve the computation of the first
two moments of $W$ only, which can often be achieved relatively easily.
An example is a random geometric network under the unit disk model.
If $\Gamma_{s,i}$ is chosen to be a neighborhood of $i$ containing
indices of all nodes whose distance to $\boldsymbol{x}_{i}$ is less
than or equal to twice the transmission range, the $b_{3}$ term is
easily shown to be $0$. It can then be readily shown that the $b_{1}$
and $b_{2}$ terms approach $0$ as the neighbourhood size of a node
becomes vanishingly small compared to the overall network size as
$\rho\rightarrow\infty$ \cite{Franceschetti06Critical}. However
this is certainly not the case for the generic random connection model
where the dependence structure is global.

Using the Chen-Stein bound, the following theorem, which summarizes
a major result of this paper can be obtained:
\begin{thm}
\label{thm:Poison convergence of isolated nodes}The number of isolated
nodes in $\mathcal{G}^{T}\left(\mathcal{X}_{\rho},g_{\rho}\right)$
converges to a Poisson distribution with mean $e^{-b}$ as $\rho\rightarrow\infty$.\end{thm}
\begin{proof}
Proof is given in the Appendix.
\end{proof}

\section{The Impact of the Boundary Effects on the Number of Isolated Nodes\label{sec:The-Impact-of-boundary effect}}

On the basis of the analysis in the last section, we now consider
the impact of the boundary effect on the number of isolated nodes
in $\mathcal{G}\left(\mathcal{X}_{\rho},g_{\rho}\right)$. Following
the same procedure that results in \eqref{eq:prob isolated node},
it can be shown that $Pr\left(I_{i}^{S}=1\right)\sim_{m}\frac{\rho}{m^{2}}e^{-\int_{A}\rho g\left(\left\Vert \frac{\boldsymbol{x}-\boldsymbol{x}_{i}}{r_{\rho}}\right\Vert \right)d\boldsymbol{x}}$
where the parameters in this section is defined analogously as those
in the last section. Note that due to the consideration of a square,
a relationship such as $\int_{A}\rho g\left(\left\Vert \frac{\boldsymbol{x}-\boldsymbol{x}_{i}}{r_{\rho}}\right\Vert \right)d\boldsymbol{x}=\int_{A}\rho g\left(\left\Vert \frac{\boldsymbol{x}}{r_{\rho}}\right\Vert \right)d\boldsymbol{x}$
is no longer valid. It follows that \begin{eqnarray*}
E\left(W^{S}\right) & = & \lim_{m\rightarrow\infty}E\left(W_{m}^{s}\right)=\int_{A}\rho e^{-\int_{A}\rho g\left(\left\Vert \frac{\boldsymbol{x}-\boldsymbol{y}}{r_{\rho}}\right\Vert \right)d\boldsymbol{x}}d\boldsymbol{y}\end{eqnarray*}
\begin{eqnarray}
\lim_{\rho\rightarrow\infty}E\left(W^{S}\right) & = & \lim_{\rho\rightarrow\infty}\int_{A_{\rho}}\rho r_{\rho}^{2}e^{-\int_{A_{\rho}}\rho r_{\rho}^{2}g\left(\left\Vert \boldsymbol{x}-\boldsymbol{y}\right\Vert \right)d\boldsymbol{x}}d\boldsymbol{y}\nonumber \\
 & = & \lim_{\rho\rightarrow\infty}\rho e^{-C\rho r_{\rho}^{2}}=e^{-b}\label{eq: asymptotic expected number of isolated nodes square}\end{eqnarray}
where $A_{\rho}$ is a square of size $\frac{1}{r_{\rho}^{2}}$ and
$A_{\rho}\triangleq\left[-\frac{1}{2r_{\rho}},\frac{1}{2r_{\rho}}\right)^{2}$.
In arriving at \eqref{eq: asymptotic expected number of isolated nodes square}
some discussions involving dividing $A_{\rho}$ into three non-overlapping
regions: four square areas of size $r_{\rho}^{-\varepsilon}\times r_{\rho}^{-\varepsilon}$
at the corners of $A_{\rho}$, denoted by $\angle A_{\rho}$; four
rectangular areas of size $r_{\rho}^{-\varepsilon}\times\left(r_{\rho}^{-1}-2r_{\rho}^{-\varepsilon}\right)$
adjacent to the four sides of $A_{\rho}$, denoted by $\ell A_{\rho}$;
and the rest central area, are omitted due to space limitation, where
$\varepsilon$ is a small positive constant and $\varepsilon<\frac{1}{4}$. 

Comparing \eqref{eq:Expected number of isolated nodes asymptotic}
and \eqref{eq: asymptotic expected number of isolated nodes square},
it is noted that the expected numbers of isolated nodes on a torus
and on a square respectively asymptotically converge to the same nonzero
finite \emph{constant} $e^{-b}$ as $\rho\rightarrow\infty$. Now
we use the coupling technique \cite{Franceschetti07Random} to construct
the connection between $W^{S}$ and $W^{T}$. Consider an instance
of $\mathcal{G}^{T}\left(\mathcal{X}_{\rho},g_{\rho}\right)$. The
number of isolated nodes in that network is $W^{T}$, which depends
on $\rho$. Remove each connection of the above network with probability
$1-\frac{g\left(\frac{x}{r_{\rho}}\right)}{g\left(\frac{x^{T}}{r_{\rho}}\right)}$,
independent of the event that another connection is removed, where
$x$ is the Euclidean distance between the two endpoints of the connection
and $x^{T}$ is the corresponding toroidal distance. Due to \eqref{eq:property of toroidal distance 1}
and the non-increasing property of $g$, $0\leq1-\frac{g\left(\frac{x}{r_{\rho}}\right)}{g\left(\frac{x^{T}}{r_{\rho}}\right)}\leq1$.
Further note that only connections between nodes near the boundary
with $x^{T}<x$ will be affected. Denote the number of \emph{newly
}appeared isolated nodes by $W^{E}$. $W^{E}$ has the meaning of
being \emph{the number of isolated nodes due to the boundary effect}.
It is straightforward to show that $W^{E}$ is a non-negative random
integer, depending on $\rho$. Further, such a connection removal
process results in a random network with nodes Poissonly distributed
with density $\rho$ where a pair of nodes separated by an \emph{Euclidean}
distance $x$ are directly connected with probability $g\left(\frac{x}{r_{\rho}}\right)$,
i.e. a random network on a square with the boundary effect included.
The following equation results as a consequence of the above discussion:
$W^{S}=W^{E}+W^{T}$. Using \eqref{eq:Expected number of isolated nodes asymptotic},
\eqref{eq: asymptotic expected number of isolated nodes square} and
the above equation, it can be shown that $\lim_{\rho\rightarrow\infty}E\left(W^{E}\right)=0$.
Due to the non-negativeness of $W^{E}$: $\lim_{\rho\rightarrow\infty}\Pr\left(W^{E}=0\right)=1$.
The above discussion is summarized in the following lemma, which forms
the second major contribution of this paper.
\begin{thm}
\label{thm:Isolated nodes due to boundary effect}The number of isolated
nodes in $\mathcal{G}\left(\mathcal{X}_{\rho},g_{\rho}\right)$ due
to the boundary effect is a.a.s. $0$ as $\rho\rightarrow\infty$.
\end{thm}

\section{The Necessary Condition for Asymptotically Connected Networks\label{sec:The-Necessary-Condition}}

We are now ready to present the necessary condition for $\mathcal{G}\left(\mathcal{X}_{\rho},g_{\rho}\right)$
to be a.a.s. connected as $\rho\rightarrow\infty$. The following
theorem can be obtained using Theorems \ref{thm:Poison convergence of isolated nodes}
and \ref{thm:Isolated nodes due to boundary effect}:
\begin{thm}
\label{thm:Number of isolated nodes square}The number of isolated
nodes in $\mathcal{G}\left(\mathcal{X}_{\rho},g_{\rho}\right)$ converges
to a Poisson distribution with mean $e^{-b}$ as $\rho\rightarrow\infty$.
\end{thm}
Corollary \ref{cor:Prob no isolated node} follows immediately from
Theorem \ref{thm:Number of isolated nodes square}.
\begin{cor}
\label{cor:Prob no isolated node}As $\rho\rightarrow\infty$, the
probability that there is no isolated node in $\mathcal{G}\left(\mathcal{X}_{\rho},g_{\rho}\right)$
converges to $e^{-e^{-b}}$.
\end{cor}
With a slight modification of the proof of Theorem \ref{thm:Poison convergence of isolated nodes},
it can be shown that Theorems \ref{thm:Poison convergence of isolated nodes}
and \ref{thm:Number of isolated nodes square} and Corollary \ref{cor:Prob no isolated node}
can be extended to the situation when $b$ is a function of $\rho$
and $\lim_{\rho\rightarrow\infty}b=B$, where $B$ is a constant.
Now we further relax the condition in Theorem \ref{thm:Poison convergence of isolated nodes}
on $b$ and consider the situation when $b\rightarrow-\infty$ or
$b\rightarrow\infty$ as $\rho\rightarrow\infty$. When $b\rightarrow\infty$,
the number of connections in $\mathcal{G}\left(\mathcal{X}_{\rho},g_{\rho}\right)$
increases. Unsurprisingly, isolated nodes disappear. In fact, using
the coupling technique, Lemma \ref{lem:Expected Isolated nodes torus},
Theorem \ref{thm:Isolated nodes due to boundary effect} and Markov's
inequality, it can be shown that if $b\rightarrow\infty$ as $\rho\rightarrow\infty$,
$\lim_{\rho\rightarrow\infty}\Pr\left(W^{S}=0\right)=1$.

Now we consider the situation when $b\rightarrow-\infty$ as $\rho\rightarrow\infty$.
For an arbitrary network, a particular property is \emph{increasing}
if the property is preserved when more connections (edges) are added
into the network. A property is \emph{decreasing} if its complement
is increasing, or equivalently a decreasing property is preserved
when connections (edges) are removed from the network. It follows
that the property that the network $\mathcal{G}\left(\mathcal{X}_{\rho},g_{\rho}\right)$
has at least one isolated node, denoted by $\Lambda$, is a \emph{decreasing}
property. The complement of $\Lambda$, denoted by $\Lambda^{c}$,
viz. the property that the network $\mathcal{G}\left(\mathcal{X}_{\rho},g_{\rho}\right)$
has no isolated node, is an increasing property. In fact the network
$\mathcal{G}_{1}\left(\mathcal{X}_{\rho},g_{\rho}\right)$ where $b=B_{1}$
can be obtained from the network $\mathcal{G}_{2}\left(\mathcal{X}_{\rho},g_{\rho}\right)$
where $b=B_{2}$ and $B_{2}<B_{1}$ by removing each connection in
$\mathcal{G}_{1}\left(\mathcal{X}_{\rho},g_{\rho}\right)$ independently
with a probability $g\left(\frac{x}{\sqrt{\frac{\log\rho+B_{2}}{C\rho}}}\right)/g\left(\frac{x}{\sqrt{\frac{\log\rho+B_{1}}{C\rho}}}\right)$
with $x$ being the  distance between two endpoints of the connection.
The above observations, together with Corollary \ref{cor:Prob no isolated node},
lead to the conclusion that if $b\rightarrow-\infty$ as $\rho\rightarrow\infty$,
\[
\lim_{\rho\rightarrow\infty}Pr\left(\Lambda\right)=\lim_{\rho\rightarrow\infty}1-Pr\left(\Lambda^{c}\right)=1\]

The above discussions are summarized in the following theorem and
corollary, which form the third major contribution of this paper:
\begin{thm}
\label{thm:isolated nodes b=00003D0 or infinity}In the network $\mathcal{G}\left(\mathcal{X}_{\rho},g_{\rho}\right)$,
if $b\rightarrow\infty$ as $\rho\rightarrow\infty$, a.a.s. there
is no isolated node in the network; if $b\rightarrow-\infty$ as $\rho\rightarrow\infty$,
a.a.s. the network has at least one isolated nodes.\end{thm}
\begin{cor}
\label{cor:necessary condition for asymptotically connected network}$b\rightarrow\infty$
is a necessary condition for the network $\mathcal{G}\left(\mathcal{X}_{\rho},g_{\rho}\right)$
to be a.a.s. connected as $\rho\rightarrow\infty$.
\end{cor}

\section{Related Work\label{sec:Related-Work}}

Extensive research has been done on connectivity problems using the
well-known random geometric graph and the \emph{unit disk model},
which is usually obtained by randomly and uniformly distributing $n$
vertices in a given area and connecting any two vertices iff their
 distance is smaller than or equal to a given threshold $r(n)$ \cite{Penrose99On,Penrose03Random}.
Significant outcomes have been achieved for both asymptotically infinite
$n$ \cite{Gupta98Critical,Xue04The,Philips89Connectivity,Ravelomanana04Extremal,Balister05Connectivity,Wan04Asymptotic,Penrose03Random,Balister09A}
and finite $n$ \cite{Bettstetter04On,Bettstetter02On,Tang03An}.
Specifically, it was shown that under the unit disk model and in $\Re^{2}$,
the above network with $r\left(n\right)=\sqrt{\frac{\log n+c\left(n\right)}{\pi n}}$
is a.a.s. connected as $n\rightarrow\infty$ iff $c\left(n\right)\rightarrow\infty$
. In \cite{Ravelomanana04Extremal}, Ravelomanana investigated the
critical transmission range for connectivity in 3-dimensional wireless
sensor networks and derived similar results as the 2-dimensional results
in \cite{Gupta98Critical}. Note that most of the results for finite
$n$ are empirical results. 

In \cite{Hekmat06Connectivity,Orriss03Probability,Miorandi05Coverage,Miorandi08The,Bettstetter04failure,Bettstetter05Connectivity}
the necessary condition for the above network to be asymptotically
connected is investigated under the more realistic \emph{log-normal
connection model}. Under the log-normal connection model, two nodes
are directly connected if the received power at one node from the
other node, whose attenuation follows the log-normal model, is greater
than a given threshold. These analysis however all relies on the \emph{assumption}
that the event that a node is isolated and the event that another
node is isolated are independent. Realistically however, one may expect
the above two events to be correlated whenever there is a non-zero
probability that a third node may exist which may have direct connections
to both nodes. In the unit disk model, this may happen when the transmission
range of the two nodes overlaps. In the log-normal model, \emph{any}
node may have a non-zero probability of having direct connections
to both nodes. This observation and the lack of rigorous analysis
on the node isolation events to support the independence assumption
raised a question mark over the validity of the results of \cite{Hekmat06Connectivity,Orriss03Probability,Miorandi05Coverage,Miorandi08The,Bettstetter04failure,Bettstetter05Connectivity}. 

The results in this paper complement the above studies in two ways.
They provide the asymptotic distribution of the number of isolated
nodes in the network, which is valid not only for the unit disk model
and the log-normal connection model but also for the more generic
random connection model. Second they do \emph{not} depend on the independence
assumption concerning isolated nodes just mentioned. In fact, it is
an unjustifiable assumption. They do however rely on the independence
of connections of different node pairs, referred to in the discussion
of the random connection model in Section \ref{sec:Introduction}. 

Some work exists on the analysis of the asymptotic distribution of
the number of isolated nodes \cite{Yi06Asymptotic,Franceschetti06Critical,Franceschetti07Random,Penrose03Random}
under the assumption of a unit disk model. In \cite{Yi06Asymptotic},
Yi et al. considered a total of $n$ nodes distributed independently
and uniformly in a unit-area disk. Using some complicated geometric
analysis, they showed that if all nodes have a maximum transmission
range $r(n)=\sqrt{\left(\log n+\xi\right)/\pi n}$ for some constant
$\xi$, the total number of isolated nodes is asymptotically Poissonly
distributed with mean $e^{-\xi}$. In \cite{Franceschetti06Critical,Franceschetti07Random},
Franceschetti et al. derived essentially the same result using the
Chen-Stein technique. A similar result can also be found in \cite{Penrose03Random}
in a continuum percolation setting. There is a major challenge in
analyzing the distribution of the number of isolated nodes \emph{under
the random connection model}; under the unit disk model, the dependence
structure is {}``local'', i.e. the event that a node is isolated
and the event that another node is isolated are dependent iff the
distance between the two nodes is smaller than twice the transmission
range, whereas under the random connection model, the dependence structure
becomes {}``global'', i.e. the above two events are dependent even
if the two nodes are far away.

\section{Conclusions and Further Work\label{sec:Conclusions-and-Further}}

In this paper, we analyzed the asymptotic distribution of the number
of isolated nodes in $\mathcal{G}\left(\mathcal{X}_{\rho},g_{\rho}\right)$
using the Chen-Stein technique, the impact of the boundary effect
on the number of isolated nodes and on that basis the necessary condition
for $\mathcal{G}\left(\mathcal{X}_{\rho},g_{\rho}\right)$ to be a.a.s.
connected as $\rho\rightarrow\infty$. Considering one instance of
such a network and expanding the distances between all pairs of nodes
by a factor of $1/r_{\rho}$ while maintaining their connections,
there results a random network with nodes Poissonly distributed on
a square of size $1/r_{\rho}^{2}$ with density $\rho r_{\rho}^{2}$
where a pair of nodes separated by an Euclidean distance $x$ are
directly connected with probability $g\left(x\right)$. Using the
scaling technique \cite{Franceschetti07Random}, it can be readily
shown that our result applies to this random network. By proper scaling
or slight modifications of the proof of Theorem \ref{thm:Poison convergence of isolated nodes},
our result can be extended to networks of other sizes.

It can be easily shown that as $\rho\rightarrow\infty$, the average
node degree in $\mathcal{G}\left(\mathcal{X}_{\rho},g_{\rho}\right)$
converges to $\log\rho+b$. That is, the average node degree under
the random connection model increases at the same rate as the average
node degree required for a connected network under the unit disk model
as $\rho\rightarrow\infty$ \cite{Philips89Connectivity}. Further
if $b\rightarrow\infty$ as $\rho\rightarrow\infty$, a.a.s. there
is no isolated node in the network. This result coincides with the
result in \cite{Gupta98Critical} on the critical transmission range
required for an a.a.s. connected network. Another implication of our
result is that different channel models appear to play little role
in determining the \emph{asymptotic distribution} of isolated nodes
(hence the connectivity) so long as they achieve the same average
node degree under the same node density. 

This paper focuses on a necessary condition for $\mathcal{G}\left(\mathcal{X}_{\rho},g_{\rho}\right)$
to be a.a.s. connected. We expect that as $\rho\rightarrow\infty$,
the necessary condition also becomes sufficient, i.e. the network
becomes connected when the last isolated node disappears. It is part
of our future work to investigate the sufficient condition for asymptotically
connected networks under the random connection model and validate
the above conjecture. 

This paper focuses on the asymptotic distribution of the number of
isolated nodes, i.e. the number of nodes with a node degree $k=0$.
We conjecture that for a generic $k$, the asymptotic distribution
of the number of nodes with degree $k$ may also converge to a Poisson
distribution. Thus, it is another direction of our future work to
examine the asymptotic distribution of the number of nodes with degree
$k$, where $k>0$.

\section*{Appendix: Proof of Theorem \ref{thm:Poison convergence of isolated nodes}\label{app:Proof-of-Theorem}}

In this appendix, we give a proof of Theorem \ref{thm:Poison convergence of isolated nodes}
using the Chen-Stein bound in Theorem \ref{thm:Chen-Stein Bound}.
The key idea involved using Theorem \ref{thm:Chen-Stein Bound} to
prove Theorem \ref{thm:Poison convergence of isolated nodes} is constructing
a neighborhood of a node, i.e. $\Gamma_{s,i}$ in Theorem \ref{thm:Chen-Stein Bound},
such that a) the size of the neighborhood becomes vanishingly small
compared with $A$ as $\rho\rightarrow\infty$. This is required for
the $b_{1}$ and $b_{2}$ terms to approach $0$ as $\rho\rightarrow\infty$;
b) a.a.s. the neighborhood contains all nodes that may have a direct
connection with the node. This is required for the $b_{3}$ term to
approach $0$ as $\rho\rightarrow\infty$. Such a neighborhood is
defined in the next paragraph.

First note that parameter $W$ in Theorem \ref{thm:Chen-Stein Bound}
has the same meaning of $W_{m}^{T}$ defined in Section \ref{sec:Isolated nodes torus}.
Therefore the parameter $\lambda$ in Theorem \ref{thm:Chen-Stein Bound},
which depends on both $\rho$ and $m$, satisfies $\lim_{\rho\rightarrow\infty}\lim_{m\rightarrow\infty}\lambda=e^{-b}$.
Further $p_{i}\triangleq E\left(I_{i}^{T}\right)$ and $E\left(I_{i}^{T}\right)$
has been given in \eqref{eq:prob isolated node}. Unless otherwise
specified, these parameters, e.g. $\boldsymbol{x}_{i}$, $m$, $I_{i}^{T}$,
$\Gamma$ and $r_{\rho}$, have the same meaning as those defined
in Section \ref{sec:Isolated nodes torus}. Denote by $D\left(\boldsymbol{x}_{i},r\right)$
a disk centered at $\boldsymbol{x}_{i}$ and with a radius $r$. Further
define the neighbourhood of an index $i\in\Gamma$ as $\Gamma_{s,i}\triangleq\left\{ j:\boldsymbol{x}_{j}\in D\left(\boldsymbol{x}_{i},2r_{\rho}^{1-\epsilon}\right)\right\} \backslash\{i\}$
and define the non-neighbourhood of the index $i$ as $\Gamma_{w,i}\triangleq\left\{ j:\boldsymbol{x}_{j}\notin D\left(\boldsymbol{x}_{i},2r_{\rho}^{1-\epsilon}\right)\right\} $
where $\epsilon$ is a constant and $\epsilon\in\left(0,\frac{1}{2}\right)$.
It can be shown that

\begin{equation}
\left|\Gamma_{s,i}\right|=m^{2}4\pi r_{\rho}^{2-2\epsilon}+o_{m}\left(m^{2}4\pi r_{\rho}^{2-2\epsilon}\right)\label{eq:tao s, i}\end{equation}

From \eqref{eq:prob isolated node}, $p_{i}=E\left(I_{i}^{T}\right)$
and \eqref{eq:Expected number of isolated nodes asymptotic}, it follows
that\begin{eqnarray}
\lim_{m\rightarrow\infty}m^{2}p_{i} & = & \rho e^{-\int_{A}\rho g\left(\frac{\left\Vert \boldsymbol{x}-\boldsymbol{x}_{i}\right\Vert ^{T}}{r_{\rho}}\right)d\boldsymbol{x}}\label{eq:value of m2p_i finite rho}\\
\lim_{\rho\rightarrow\infty}\lim_{m\rightarrow\infty}m^{2}p_{i} & = & e^{-b}\label{eq:limiting value m2p_i}\end{eqnarray}

Next we shall evaluate the $b_{1}$, $b_{2}$ and $b_{3}$ terms separately.

\subsection{An Evaluation of the $b_{1}$ Term\label{sub:An-Evaluation-of-b1}}

It can be shown that\begin{eqnarray}
 &  & \lim_{m\rightarrow\infty}\sum_{i\in\Gamma}\left(p_{i}^{2}+p_{i}E\left(\sum_{j\in\Gamma_{s,i}}I_{j}^{T}\right)\right)\nonumber \\
 & = & \lim_{m\rightarrow\infty}m^{2}p_{i}E\left(\sum_{j\in\Gamma_{s,i}\cup\left\{ i\right\} }I_{j}^{T}\right)\nonumber \\
 & = & \lim_{m\rightarrow\infty}\left(m^{2}p_{i}\right)^{2}4\pi r_{\rho}^{2-2\epsilon}\nonumber \\
 & = & 4\pi\left(\rho e^{-\int_{A}\rho g\left(\frac{\left\Vert \boldsymbol{x}-\boldsymbol{x}_{i}\right\Vert ^{T}}{r_{\rho}}\right)d\boldsymbol{x}}\right)^{2}\left(\frac{\log\rho+b}{C\rho}\right)^{1-\epsilon}\label{eq:evaluation-b1-finite-rho}\end{eqnarray}
where in the second step, \eqref{eq:tao s, i} is used and in the
final step \eqref{eq:prob isolated node}, \eqref{eq:value of m2p_i finite rho}
and the value of $r_{\rho}$ are used. It follows that\[
\lim_{\rho\rightarrow\infty}\textrm{RHS of}\;\eqref{eq:evaluation-b1-finite-rho}=4\pi e^{-2b}\lim_{\rho\rightarrow\infty}\left(\frac{\log\rho+b}{C\rho}\right)^{1-\epsilon}=0\]
where \eqref{eq:expected number of isolated nodes finite network}
and \eqref{eq:Expected number of isolated nodes asymptotic} are used
in the above equation, and RHS is short for the right hand side. This
leads to the conclusion that $\lim_{\rho\rightarrow\infty}\lim_{m\rightarrow\infty}b_{1}=0$.

\subsection{An Evaluation of the $b_{2}$ Term\label{sub:An-Evaluation-of-b2}}

For the $b_{2}$ term, we observe that

\begin{eqnarray}
 &  & \lim_{m\rightarrow\infty}\sum_{i\in\Gamma}E\left(I_{i}^{T}\sum_{j\in\Gamma_{s,i}}I_{j}^{T}\right)\nonumber \\
 & = & \lim_{m\rightarrow\infty}\frac{\rho^{2}}{m^{2}}\sum_{j\in\Gamma_{s,i}}\left\{ \left(1-g\left\Vert \frac{\boldsymbol{x}_{i}-\boldsymbol{x}_{j}}{r_{\rho}}\right\Vert ^{T}\right)\right.\nonumber \\
 & \times & \exp\left[-\int_{A}\rho\left(g\left(\left\Vert \frac{\boldsymbol{x}-\boldsymbol{x}_{i}}{r_{\rho}}\right\Vert ^{T}\right)+g\left(\left\Vert \frac{\boldsymbol{x}-\boldsymbol{x}_{j}}{r_{\rho}}\right\Vert ^{T}\right)\right.\right.\nonumber \\
 &  & \left.\left.\left.-g\left(\left\Vert \frac{\boldsymbol{x}-\boldsymbol{x}_{i}}{r_{\rho}}\right\Vert ^{T}\right)g\left(\left\Vert \frac{\boldsymbol{x}-\boldsymbol{x}_{j}}{r_{\rho}}\right\Vert ^{T}\right)\right)d\boldsymbol{x}\right]\right\} \nonumber \\
 & = & \rho^{2}\int_{D\left(\boldsymbol{x}_{i},2r_{\rho}^{1-\epsilon}\right)}\left\{ \left(1-g\left(\frac{\left\Vert \boldsymbol{x}_{i}-\boldsymbol{y}\right\Vert ^{T}}{r_{\rho}}\right)\right)\right.\nonumber \\
 & \times & \exp\left[-\int_{A}\rho\left(g\left(\left\Vert \frac{\boldsymbol{x}-\boldsymbol{x}_{i}}{r_{\rho}}\right\Vert ^{T}\right)+g\left(\left\Vert \frac{\boldsymbol{x}-\boldsymbol{y}}{r_{\rho}}\right\Vert ^{T}\right)\right.\right.\nonumber \\
 &  & \left.\left.\left.-g\left(\left\Vert \frac{\boldsymbol{x}-\boldsymbol{x}_{i}}{r_{\rho}}\right\Vert ^{T}\right)g\left(\left\Vert \frac{\boldsymbol{x}-\boldsymbol{y}}{r_{\rho}}\right\Vert ^{T}\right)\right)d\boldsymbol{x}\right]\right\} d\boldsymbol{y}\nonumber \\
 & = & \rho^{2}r_{\rho}^{2}\int_{D\left(\boldsymbol{0},2r_{\rho}^{-\epsilon}\right)}\left\{ \left(1-g\left(\left\Vert \boldsymbol{y}\right\Vert ^{T}\right)\right)\right.\nonumber \\
 & \times & \exp\left[-\rho r_{\rho}^{2}\int_{A_{\rho}}\left(g\left(\left\Vert \boldsymbol{x}\right\Vert ^{T}\right)+g\left(\left\Vert \boldsymbol{x}-\boldsymbol{y}\right\Vert ^{T}\right)\right.\right.\nonumber \\
 &  & \left.\left.\left.-g\left(\left\Vert \boldsymbol{x}\right\Vert ^{T}\right)g\left(\left\Vert \boldsymbol{x}-\boldsymbol{y}\right\Vert ^{T}\right)\right)d\boldsymbol{x}\right]\right\} d\boldsymbol{y}\label{eq:Value-of-b2-finite-rho}\end{eqnarray}
where $A_{\rho}=\left[-\frac{1}{2r_{\rho}},\frac{1}{2r_{\rho}}\right)^{2}$,
in the first step, \eqref{eq:Prob of having a node in S_i}, \eqref{eq:joint distribution of isolated node events}
and \eqref{eq:Joint distribution of I_i and I_j} are used and the
final step involves some translation and scaling operations. Let $\lambda\triangleq\frac{\log\rho+b}{C}$,
it can be further shown that as $\rho\rightarrow\infty$, 

\begin{eqnarray}
 &  & e^{2b}\lim_{\rho\rightarrow\infty}\textrm{RHS of}\;\eqref{eq:Value-of-b2-finite-rho}\nonumber \\
 & \leq & \lim_{\rho\rightarrow\infty}\frac{\lambda}{\rho}\int_{D\left(\boldsymbol{0},2r_{\rho}^{-\epsilon}\right)}e^{\lambda\int_{\Re^{2}}g\left(\left\Vert \boldsymbol{x}\right\Vert ^{T}\right)g\left(\left\Vert \boldsymbol{x}-\boldsymbol{y}\right\Vert ^{T}\right)d\boldsymbol{x}}d\boldsymbol{y}\nonumber \\
 & = & \lim_{\rho\rightarrow\infty}\frac{\log\rho}{C\rho}\int_{D\left(\boldsymbol{0},2r_{\rho}^{-\epsilon}\right)}e^{\lambda\int_{\Re^{2}}h\left(\boldsymbol{x},\boldsymbol{y}\right)d\boldsymbol{x}}d\boldsymbol{y}\nonumber \\
 & = & \lim_{\rho\rightarrow\infty}\frac{1}{C\rho}\left\{ \int_{D\left(\boldsymbol{0},2r_{\rho}^{-\epsilon}\right)}e^{\lambda\int_{\Re^{2}}h\left(\boldsymbol{x},\boldsymbol{y}\right)d\boldsymbol{x}}d\boldsymbol{y}\right.\nonumber \\
 & + & \frac{\log\rho\left(\log\rho+b-1\right)}{C\rho}4\pi\epsilon r_{\rho}^{-2\epsilon-2}\nonumber \\
 & \times & e^{\frac{\log\rho+b}{C}\int_{\Re^{2}}g\left(\left\Vert \boldsymbol{x}\right\Vert ^{T}\right)g\left(\left\Vert \boldsymbol{x}-2r_{\rho}^{-\epsilon}\boldsymbol{u}\right\Vert ^{T}\right)d\boldsymbol{x}}\nonumber \\
 & + & \int_{D\left(\boldsymbol{0},2r_{\rho}^{-\epsilon}\right)}\left[e^{\frac{\log\rho+b}{C}\int_{\Re^{2}}h\left(\boldsymbol{x},\boldsymbol{y}\right)d\boldsymbol{x}}\right.\nonumber \\
 & \times & \left.\left.\frac{\log\rho\int_{\Re^{2}}h\left(\boldsymbol{x},\boldsymbol{y}\right)d\boldsymbol{x}}{C}\right]d\boldsymbol{y}\right\} \label{eq:b2 term intermediate step 2}\end{eqnarray}
where $\boldsymbol{u}$ is a unit vector pointing to the $+x$ direction
and $h\left(\boldsymbol{x},\boldsymbol{y}\right)=g\left(\left\Vert \boldsymbol{x}\right\Vert ^{T}\right)g\left(\left\Vert \boldsymbol{x}-\boldsymbol{u}\left\Vert \boldsymbol{y}\right\Vert ^{T}\right\Vert ^{T}\right)$,
in the first step \eqref{eq:expected number of isolated nodes finite network},
\eqref{eq:Expected number of isolated nodes asymptotic}, $r_{\rho}=\sqrt{\frac{\log\rho+b}{C\rho}}$
and $1-g\left(\left\Vert \boldsymbol{y}\right\Vert ^{T}\right)\leq1$
are used, and in the last step, L'H�pital's rule, where $C\rho$ is
used as the denominator and the other terms are used as the numerator,
\eqref{eq:property of toroidal distance 1} and the following formulas
are used: \begin{eqnarray*}
 &  & \frac{d}{dx}\int_{0}^{h(x)}f\left(x,y\right)dy\\
 & = & \int_{0}^{h\left(x\right)}\frac{\partial f\left(x,y\right)}{\partial x}dy+f\left(x,h\left(x\right)\right)\frac{dh\left(x\right)}{dx}\end{eqnarray*}
\[
\frac{d}{d\rho}\left(r_{\rho}^{-2\epsilon}\right)=\epsilon r_{\rho}^{-2\epsilon-2}\frac{\log\rho+b-1}{C\rho^{2}}\]

In the following we show that all three terms inside the $\lim_{\rho\rightarrow\infty}$
sign and separated by $+$ sign in $ $ \eqref{eq:b2 term intermediate step 2}
approach $0$ as $\rho\rightarrow\infty$. First it can be shown that\begin{eqnarray}
 &  & \int_{\Re^{2}}g\left(\left\Vert \boldsymbol{x}\right\Vert \right)g\left(\left\Vert \boldsymbol{x}-\boldsymbol{u}2r_{\rho}^{-\epsilon}\right\Vert \right)d\boldsymbol{x}\nonumber \\
 & = & \int_{D\left(\boldsymbol{0},r_{\rho}^{-\epsilon}\right)}g\left(\left\Vert \boldsymbol{x}\right\Vert \right)g\left(\left\Vert \boldsymbol{x}-\boldsymbol{u}2r_{\rho}^{-\epsilon}\right\Vert \right)d\boldsymbol{x}\nonumber \\
 & + & \int_{\Re^{2}\backslash D\left(\boldsymbol{0},r_{\rho}^{-\epsilon}\right)}g\left(\left\Vert \boldsymbol{x}\right\Vert \right)g\left(\left\Vert \boldsymbol{x}-\boldsymbol{u}2r_{\rho}^{-\epsilon}\right\Vert \right)d\boldsymbol{x}\nonumber \\
 & \leq & \int_{D\left(\boldsymbol{0},r_{\rho}^{-\epsilon}\right)}g\left(\left\Vert \boldsymbol{x}\right\Vert \right)g\left(r_{\rho}^{-\epsilon}\right)d\boldsymbol{x}\nonumber \\
 & + & \int_{\Re^{2}\backslash D\left(\boldsymbol{0},r_{\rho}^{-\epsilon}\right)}g\left(r_{\rho}^{-\epsilon}\right)g\left(\left\Vert \boldsymbol{x}-\boldsymbol{u}2r_{\rho}^{-\epsilon}\right\Vert \right)d\boldsymbol{x}\nonumber \\
 & \leq & 2Cg\left(r_{\rho}^{-\epsilon}\right)=o_{\rho}\left(r_{\rho}^{2\epsilon}\right)\label{eq:intermediate step in the first term in b2}\end{eqnarray}
where in the second step the observation that the distance between
any point in $D\left(\boldsymbol{0},r_{\rho}^{-\epsilon}\right)$
and $\boldsymbol{u}2r_{\rho}^{-\epsilon}$ is larger than or equal
to $r_{\rho}^{-\epsilon}$, the observation that the distance between
any point in $\Re^{2}\backslash D\left(\boldsymbol{0},r_{\rho}^{-\epsilon}\right)$
and the origin is larger than or equal to $r_{\rho}^{-\epsilon}$
and the non-increasing property of $g$ are used, \eqref{eq:scaling property of g(x)}
is used in the last step. This readily leads to the result that the
first term in \eqref{eq:b2 term intermediate step 2} satisfies: \begin{eqnarray*}
 &  & \lim_{\rho\rightarrow\infty}\frac{1}{C\rho}\int_{D\left(\boldsymbol{0},2r_{\rho}^{-\epsilon}\right)}e^{\frac{\log\rho+b}{C}\int_{\Re^{2}}h\left(\boldsymbol{x},\boldsymbol{y}\right)d\boldsymbol{x}}d\boldsymbol{y}\\
 & = & \lim_{\rho\rightarrow\infty}\left[e^{\frac{\log\rho+b}{C}\int_{\Re^{2}}g\left(\left\Vert \boldsymbol{x}\right\Vert ^{T}\right)g\left(\left\Vert \boldsymbol{x}-\boldsymbol{u}2r_{\rho}^{-\epsilon}\right\Vert ^{T}\right)d\boldsymbol{x}}\right.\\
 & \times & 4\pi\epsilon r_{\rho}^{-2\epsilon-2}\frac{\log\rho+b-1}{C^{2}\rho^{2}}\\
 & + & \left.\int_{D\left(\boldsymbol{0},2r_{\rho}^{-\epsilon}\right)}e^{\frac{\log\rho+b}{C}\int_{\Re^{2}}h\left(\boldsymbol{x},\boldsymbol{y}\right)d\boldsymbol{x}}\frac{\int_{\Re^{2}}h\left(\boldsymbol{x},\boldsymbol{y}\right)d\boldsymbol{x}}{C^{2}\rho}d\boldsymbol{y}\right]\\
 & = & \lim_{\rho\rightarrow\infty}\left[4\pi\epsilon r_{\rho}^{-2\epsilon-2}\frac{\log\rho+b-1}{C^{2}\rho^{2}}\right.\\
 & + & \left.\int_{D\left(\boldsymbol{0},2r_{\rho}^{-\epsilon}\right)}e^{\frac{\log\rho+b}{C}\int_{\Re^{2}}h\left(\boldsymbol{x},\boldsymbol{y}\right)d\boldsymbol{x}}\frac{\int_{\Re^{2}}h\left(\boldsymbol{x},\boldsymbol{y}\right)d\boldsymbol{x}}{C^{2}\rho}d\boldsymbol{y}\right]\\
 & = & \lim_{\rho\rightarrow\infty}\left[\int_{D\left(\boldsymbol{0},2r_{\rho}^{-\epsilon}\right)}e^{\frac{\log\rho+b}{C}\int_{\Re^{2}}h\left(\boldsymbol{x},\boldsymbol{y}\right)d\boldsymbol{x}}\frac{\int_{\Re^{2}}h\left(\boldsymbol{x},\boldsymbol{y}\right)d\boldsymbol{x}}{C^{2}\rho}d\boldsymbol{y}\right]\\
 & = & 0\end{eqnarray*}
where L'H�pital's rule, where $C\rho$ is used as the denominator
and the other terms are used as the numerator, and $r_{\rho}=\sqrt{\frac{\log\rho+b}{C\rho}}$
are used in the first step of the above equation, in the second step
\eqref{eq:intermediate step in the first term in b2} is used, which
readily leads to the conclusion that\begin{equation}
\lim_{\rho\rightarrow\infty}e^{\frac{\log\rho+b}{C}\int_{\Re^{2}}g\left(\left\Vert \boldsymbol{x}\right\Vert ^{T}\right)g\left(\left\Vert \boldsymbol{x}-\boldsymbol{u}2r_{\rho}^{-\epsilon}\right\Vert ^{T}\right)d\boldsymbol{x}}=1\label{eq:intermediate step 2 first term b2}\end{equation}
The final steps are complete by putting the value of $r_{\rho}$ into
the equation and noting that $\int_{\Re^{2}}h\left(\boldsymbol{x},\boldsymbol{y}\right)d\boldsymbol{x}<C$
for $\boldsymbol{y}\neq\boldsymbol{0}$, which is a consequence of
the following derivations:

\begin{eqnarray*}
 &  & \int_{\Re^{2}}g\left(\left\Vert \boldsymbol{x}\right\Vert ^{T}\right)g\left(\left\Vert \boldsymbol{x}-\boldsymbol{u}\left\Vert \boldsymbol{y}\right\Vert ^{T}\right\Vert ^{T}\right)d\boldsymbol{x}-C\\
 & = & \int_{\Re^{2}}g\left(\left\Vert \boldsymbol{x}\right\Vert ^{T}\right)\left(g\left(\left\Vert \boldsymbol{x}-\boldsymbol{u}\left\Vert \boldsymbol{y}\right\Vert ^{T}\right\Vert ^{T}\right)-1\right)d\boldsymbol{x}\leq0\end{eqnarray*}
and the only possibility for $\int_{\Re^{2}}h\left(\boldsymbol{x},\boldsymbol{y}\right)d\boldsymbol{x}-C=0$
to occur is when $g$ corresponds to a unit disk model \emph{and}
$\boldsymbol{y}=\boldsymbol{0}$.

For the second term in \eqref{eq:b2 term intermediate step 2}, it
can be shown that

\begin{eqnarray*}
 &  & \lim_{\rho\rightarrow\infty}\left[4\pi\epsilon r_{\rho}^{-2\epsilon-2}\frac{\log\rho\left(\log\rho+b-1\right)}{C^{2}\rho^{2}}\right.\\
 & \times & \left.e^{\frac{\log\rho+b}{C}\int_{\Re^{2}}g\left(\left\Vert \boldsymbol{x}\right\Vert ^{T}\right)g\left(\left\Vert \boldsymbol{x}-2\boldsymbol{u}r_{\rho}^{-\epsilon}\right\Vert ^{T}\right)d\boldsymbol{x}}\right]\\
 & = & \lim_{\rho\rightarrow\infty}4\pi\epsilon r_{\rho}^{-2\epsilon-2}\frac{\log\rho\left(\log\rho+b-1\right)}{C^{2}\rho^{2}}=0\end{eqnarray*}
where in the first step \eqref{eq:intermediate step 2 first term b2}
is used. 

For the third term in \eqref{eq:b2 term intermediate step 2}, it
is observed that

\begin{eqnarray*}
 &  & \lim_{\rho\rightarrow\infty}\frac{\log\rho}{C^{2}\rho}\int_{D\left(\boldsymbol{0},2r_{\rho}^{-\epsilon}\right)}\left[e^{\frac{\log\rho+b}{C}\int_{\Re^{2}}h\left(\boldsymbol{x},\boldsymbol{y}\right)d\boldsymbol{x}}\right.\\
 & \times & \left.\int_{\Re^{2}}h\left(\boldsymbol{x},\boldsymbol{y}\right)d\boldsymbol{x}\right]d\boldsymbol{y}\\
 & \leq & \lim_{\rho\rightarrow\infty}\frac{\log\rho}{C\rho}\int_{D\left(\boldsymbol{0},2r_{\rho}^{-\epsilon}\right)}e^{\frac{\log\rho+b}{C}\int_{\Re^{2}}h\left(\boldsymbol{x},\boldsymbol{y}\right)d\boldsymbol{x}}d\boldsymbol{y}=0\end{eqnarray*}
where $\int_{\Re^{2}}h\left(\boldsymbol{x},\boldsymbol{y}\right)d\boldsymbol{x}<C$
for $\boldsymbol{y}\neq\boldsymbol{0}$ is used in the first step.

Eventually we get $\lim_{\rho\rightarrow\infty}\lim_{m\rightarrow\infty}b_{2}=0$.

\subsection{An Evaluation of the $b_{3}$ Term\label{sub:An-Evaluation-of-b3}}

Denote by $\Gamma_{i}$ a random set of indices containing all indices
$j$ where $j\in\Gamma_{w,i}$ \emph{and }$I_{j}=1$, i.e. the node
in question is also isolated, and denote by $\gamma_{i}$ an instance
of $\Gamma_{i}$. Define $n\triangleq\left|\gamma_{i}\right|$. Following
a similar procedure that leads to \eqref{eq:ratio correlation} and
\eqref{eq:relation isolated node and non-isolated} and using the
result that $\int_{A}\rho g\left(\frac{\left\Vert \boldsymbol{x}-\boldsymbol{x}_{i}\right\Vert ^{T}}{r_{\rho}}\right)g\left(\frac{\left\Vert \boldsymbol{x}-\boldsymbol{x}_{j}\right\Vert ^{T}}{r_{\rho}}\right)d\boldsymbol{x}=o_{\rho}\left(1\right)$
and $g\left(\frac{\left\Vert \boldsymbol{x}_{i}-\boldsymbol{x}_{j}\right\Vert ^{T}}{r_{\rho}}\right)=o_{\rho}\left(1\right)$
for $\left\Vert \boldsymbol{x}_{i}-\boldsymbol{x}_{j}\right\Vert ^{T}\geq2r_{\rho}^{1-\varepsilon}$
(see \eqref{eq:intermediate step 2 first term b2}), it can be shown
that\begin{eqnarray}
 &  & \lim_{\rho\rightarrow\infty}\lim_{m\rightarrow\infty}\frac{E\left\{ I_{i}^{T}\left|\left(I_{j}^{T},j\in\Gamma_{w,i}\right)\right.\right\} }{\frac{\rho}{m^{2}}}\nonumber \\
 & = & \lim_{\rho\rightarrow\infty}E\left[e^{-\int_{A}\rho g\left(\frac{\left\Vert \boldsymbol{x}-\boldsymbol{x}_{i}\right\Vert ^{T}}{r_{\rho}}\right)\prod_{j\in\gamma_{i}}\left(1-g\left(\frac{\left\Vert \boldsymbol{x}-\boldsymbol{x}_{j}\right\Vert ^{T}}{r_{\rho}}\right)\right)d\boldsymbol{x}}\right.\nonumber \\
 & \times & \left.\prod_{j\in\gamma_{i}}\left(1-g\left(\frac{\left\Vert \boldsymbol{x}_{i}-\boldsymbol{x}_{j}\right\Vert ^{T}}{r_{\rho}}\right)\right)\right]\label{eq:Conditional value I_i b3}\end{eqnarray}
Note that $\boldsymbol{x}_{i}$ and $\boldsymbol{x}_{j},j\in\Gamma_{w,i}$
is separated by a distance not smaller than $2r_{\rho}^{-\epsilon}$.
A lower bound on the value inside the expectation operator in \eqref{eq:Conditional value I_i b3}
is given by\begin{eqnarray}
B_{L,i} & \triangleq & \left(1-g\left(2r_{\rho}^{-\epsilon}\right)\right)^{n}e^{-\int_{A}\rho g\left(\frac{\left\Vert \boldsymbol{x}-\boldsymbol{x}_{i}\right\Vert ^{T}}{r_{\rho}}\right)d\boldsymbol{x}}\label{eq:b3 term lower bound and m2p_i}\end{eqnarray}
 An upper bound on the value inside the expectation operator in \eqref{eq:Conditional value I_i b3}
is given by\begin{equation}
B_{U,i}\triangleq e^{-\int_{A}\rho g\left(\frac{\left\Vert \boldsymbol{x}-\boldsymbol{x}_{i}\right\Vert ^{T}}{r_{\rho}}\right)\prod_{j\in\gamma_{i}}\left(1-g\left(\frac{\left\Vert \boldsymbol{x}-\boldsymbol{x}_{j}\right\Vert ^{T}}{r_{\rho}}\right)\right)d\boldsymbol{x}}\label{eq:definition of B_{u,i}}\end{equation}
 It can be shown that \begin{equation}
B_{U,i}\geq\lim_{m\rightarrow\infty}\frac{mp_{i}^{2}}{\rho}\geq B_{L,i}\label{eq:relation upper bound on b3 and m2p_i}\end{equation}

Let us consider $E\left|E\left\{ I_{i}\left|\left(I_{j},j\in\Gamma_{w,i}\right)\right.\right\} -p_{i}\right|$
now. From \eqref{eq:Conditional value I_i b3}, \eqref{eq:b3 term lower bound and m2p_i},
\eqref{eq:definition of B_{u,i}} and \eqref{eq:relation upper bound on b3 and m2p_i},
it is clear that \begin{eqnarray}
 &  & \lim_{\rho\rightarrow\infty}\lim_{m\rightarrow\infty}\sum_{i\in\Gamma}E\left|E\left\{ I_{i}^{T}\left|\left(I_{j}^{T},j\in\Gamma_{w,i}\right)\right.\right\} -p_{i}\right|\nonumber \\
 & \in & \left[0,\;\max\left\{ \lim_{\rho\rightarrow\infty}\lim_{m\rightarrow\infty}m^{2}p_{i}-\rho E\left(B_{L,i}\right),\right.\right.\nonumber \\
 &  & \left.\left.\lim_{\rho\rightarrow\infty}\lim_{m\rightarrow\infty}\rho E\left(B_{U,i}\right)-m^{2}p_{i}\right\} \right]\label{eq:b3 a bound on the value}\end{eqnarray}

In the following we will show that both terms $\lim_{m\rightarrow\infty}m^{2}p_{i}-\rho E\left(B_{L,i}\right)$
and $\lim_{m\rightarrow\infty}\rho E\left(B_{U,i}\right)-m^{2}p_{i}$
in \eqref{eq:b3 a bound on the value} approach $0$ as $\rho\rightarrow\infty$.
First it can be shown that

\begin{eqnarray}
 &  & \lim_{m\rightarrow\infty}\rho E\left(B_{L,i}\right)\nonumber \\
 & \geq & \lim_{m\rightarrow\infty}\rho E\left(\left(1-ng\left(2r_{\rho}^{-\epsilon}\right)\right)e^{-\int_{A}\rho g\left(\frac{\left\Vert \boldsymbol{x}-\boldsymbol{x}_{i}\right\Vert ^{T}}{r_{\rho}}\right)d\boldsymbol{x}}\right)\nonumber \\
 & = & \lim_{m\rightarrow\infty}\rho\left(1-E\left(n\right)g\left(2r_{\rho}^{-\epsilon}\right)\right)e^{-\int_{A}\rho g\left(\frac{\left\Vert \boldsymbol{x}-\boldsymbol{x}_{i}\right\Vert ^{T}}{r_{\rho}}\right)d\boldsymbol{x}}\label{eq:b3 lower bound x rho finite rho}\end{eqnarray}
where $\lim_{m\rightarrow\infty}E\left(n\right)$ is the expected
number of isolated nodes in $A\backslash D\left(\boldsymbol{x}_{i},2r_{\rho}^{1-\epsilon}\right)$.
In the first step of the above equation, the inequality $\left(1-x\right)^{n}\geq1-nx$
for $0\leq x\leq1$ and $n\geq0$ is used. When $\rho\rightarrow\infty$,
$r_{\rho}^{1-\epsilon}\rightarrow0$ and $r_{\rho}^{-\epsilon}\rightarrow\infty$
therefore $\lim_{\rho\rightarrow\infty}\lim_{m\rightarrow\infty}E\left(n\right)=\lim_{\rho\rightarrow\infty}E\left(W\right)=e^{-b}$
is a bounded value and $\lim_{\rho\rightarrow\infty}\lim_{m\rightarrow\infty}g\left(2r_{\rho}^{-\epsilon}\right)\rightarrow0$,
which is an immediate outcome of \eqref{eq:scaling property of g(x)}
. It then follows that \begin{eqnarray*}
 &  & \lim_{\rho\rightarrow\infty}\lim_{m\rightarrow\infty}\frac{\rho E\left(B_{L,i}\right)}{m^{2}p_{i}}\\
 & \geq & \lim_{\rho\rightarrow\infty}\lim_{m\rightarrow\infty}\left(1-E\left(n\right)g\left(2r_{\rho}^{-\epsilon}\right)\right)=1\end{eqnarray*}
 Together with \eqref{eq:limiting value m2p_i} and\eqref{eq:relation upper bound on b3 and m2p_i},
it follows that \begin{equation}
\lim_{\rho\rightarrow\infty}\lim_{m\rightarrow\infty}m^{2}p_{i}-\rho E\left(B_{L,i}\right)=0\label{eq:lower bound on b3 equals 0}\end{equation}
Now let us consider the second term $\lim_{m\rightarrow\infty}\rho E\left(B_{U,i}\right)-m^{2}p_{i}$,
it can be observed that

\begin{eqnarray*}
 &  & \lim_{m\rightarrow\infty}E\left(B_{U,i}\right)\\
 & \leq & E\left[e^{-\int_{D\left(\boldsymbol{x}_{i},r_{\rho}^{1-\epsilon}\right)}\left(\rho g\left(\frac{\left\Vert \boldsymbol{x}-\boldsymbol{x}_{i}\right\Vert ^{T}}{r_{\rho}}\right)\right.}\right.\\
 &  & \left.^{\left.\prod_{j\in\gamma_{i}}\left(1-g\left(\frac{\left\Vert \boldsymbol{x}-\boldsymbol{x}_{j}\right\Vert ^{T}}{r_{\rho}}\right)\right)\right)d\boldsymbol{x}}\right]\end{eqnarray*}

\begin{eqnarray}
 & \leq & \lim_{m\rightarrow\infty}E\left[e^{-\int_{D\left(\boldsymbol{x}_{i},r_{\rho}^{1-\epsilon}\right)}\left(\rho g\left(\frac{\left\Vert \boldsymbol{x}-\boldsymbol{x}_{i}\right\Vert ^{T}}{r_{\rho}}\right)\right.}\right.\nonumber \\
 &  & \left.^{\left.\prod_{j\in\gamma_{i}}\left(1-g\left(\frac{r_{\rho}^{1-\epsilon}}{r_{\rho}}\right)\right)\right)d\boldsymbol{x}}\right]\nonumber \\
 & = & \lim_{m\rightarrow\infty}E\left(e^{-\left(1-g\left(r_{\rho}^{-\epsilon}\right)\right)^{n}\int_{D\left(\boldsymbol{x}_{i},r_{\rho}^{1-\epsilon}\right)}\rho g\left(\frac{\left\Vert \boldsymbol{x}-\boldsymbol{x}_{i}\right\Vert ^{T}}{r_{\rho}}\right)d\boldsymbol{x}}\right)\nonumber \\
 & \leq & \lim_{m\rightarrow\infty}E\left(e^{-\left(1-ng\left(r_{\rho}^{-\epsilon}\right)\right)\int_{D\left(\boldsymbol{x}_{i},r_{\rho}^{1-\epsilon}\right)}\rho g\left(\frac{\left\Vert \boldsymbol{x}-\boldsymbol{x}_{i}\right\Vert ^{T}}{r_{\rho}}\right)d\boldsymbol{x}}\right)\nonumber \\
 & {}\label{eq:limiting value of the upper bound final step}\end{eqnarray}
where in the second step, the non-increasing property of $g$, and
the fact that $\boldsymbol{x}_{j}$ is located in $A\backslash D\left(\boldsymbol{x}_{i},2r_{\rho}^{1-\epsilon}\right)$
and $\boldsymbol{x}$ is located in $D\left(\boldsymbol{x}_{i},r_{\rho}^{1-\epsilon}\right)$,
therefore $\left\Vert \boldsymbol{x}-\boldsymbol{x}_{j}\right\Vert ^{T}\geq r_{\rho}^{1-\epsilon}$
is used. It can be further demonstrated, using similar steps that
result in \eqref{eq:expected number of isolated nodes finite network}
and \eqref{eq:Expected number of isolated nodes asymptotic}, that
the term $\int_{D\left(\boldsymbol{x}_{i},r_{\rho}^{1-\epsilon}\right)}\rho g\left(\frac{\left\Vert \boldsymbol{x}-\boldsymbol{x}_{i}\right\Vert }{r_{\rho}}\right)d\boldsymbol{x}$
in \eqref{eq:limiting value of the upper bound final step} have the
following property: \begin{eqnarray}
\eta\left(\varepsilon,\rho\right) & \triangleq & \int_{D\left(\boldsymbol{x}_{i},r_{\rho}^{1-\epsilon}\right)}\rho g\left(\frac{\left\Vert \boldsymbol{x}-\boldsymbol{x}_{i}\right\Vert ^{T}}{r_{\rho}}\right)d\boldsymbol{x}\nonumber \\
 & = & \rho r_{\rho}^{2}\int_{D\left(\frac{\boldsymbol{x}_{i}}{r_{\rho}},r_{\rho}^{-\epsilon}\right)}g\left(\left\Vert \boldsymbol{x}-\frac{\boldsymbol{x}_{i}}{r_{\rho}}\right\Vert ^{T}\right)d\boldsymbol{x}\nonumber \\
 & \leq & C\rho r_{\rho}^{2}=\log\rho+b\label{eq:b3 upper bound the integral}\end{eqnarray}
For the other term $ng\left(r_{\rho}^{-\epsilon}\right)$ in \eqref{eq:limiting value of the upper bound final step},
choosing a positive constant $\delta<2\epsilon$ and using Markov's
inequality, it can be shown that\[
Pr\left(n\geq r_{\rho}^{-\delta}\right)\leqslant r_{\rho}^{\delta}E\left(n\right)\]
\begin{eqnarray*}
 &  & \lim_{\rho\rightarrow\infty}\lim_{m\rightarrow\infty}Pr\left(ng\left(r_{\rho}^{-\epsilon}\right)\eta\left(\varepsilon,\rho\right)\geq r_{\rho}^{-\delta}g\left(r_{\rho}^{-\epsilon}\right)\eta\left(\varepsilon,\rho\right)\right)\\
 & \leq & \lim_{\rho\rightarrow\infty}\lim_{m\rightarrow\infty}r_{\rho}^{\delta}E\left(n\right)\end{eqnarray*}
where $\lim_{\rho\rightarrow\infty}r_{\rho}^{-\delta}g\left(r_{\rho}^{-\epsilon}\right)\eta\left(\varepsilon,\rho\right)=0$
due to \eqref{eq:scaling property of g(x)}, \eqref{eq:b3 upper bound the integral}
and $\delta<2\epsilon$, $\lim_{\rho\rightarrow\infty}r_{\rho}^{B}=0$
for any positive constant $B$, and $\lim_{\rho\rightarrow\infty}\lim_{m\rightarrow\infty}r_{\rho}^{\delta}E\left(n\right)=0$
due to that $\lim_{\rho\rightarrow\infty}\lim_{m\rightarrow\infty}E\left(n\right)=\lim_{\rho\rightarrow\infty}E\left(W\right)=e^{-b}$
is a bounded value and that $\lim_{\rho\rightarrow\infty}r_{\rho}^{\delta}=0$.
Therefore\begin{equation}
\lim_{\rho\rightarrow\infty}\lim_{m\rightarrow\infty}Pr\left(ng\left(r_{\rho}^{-\epsilon}\right)\eta\left(\varepsilon,\rho\right)=0\right)=1\label{eq:asymptotic almost surely expected n}\end{equation}

As a result of \eqref{eq:limiting value of the upper bound final step},
\eqref{eq:asymptotic almost surely expected n}, \eqref{eq:b3 upper bound the integral},
\eqref{eq:expected number of isolated nodes finite network} and \eqref{eq:Expected number of isolated nodes asymptotic}:
\begin{eqnarray*}
 &  & \lim_{\rho\rightarrow\infty}\lim_{m\rightarrow\infty}\rho E\left(B_{U,i}\right)\\
 & \leq & \lim_{\rho\rightarrow\infty}\lim_{m\rightarrow\infty}\rho E\left(e^{-\int_{D\left(\boldsymbol{x}_{i},r_{\rho}^{1-\epsilon}\right)}\rho g\left(\frac{\left\Vert \boldsymbol{x}-\boldsymbol{x}_{i}\right\Vert ^{T}}{r_{\rho}}\right)d\boldsymbol{x}}\right)\\
 & = & \lim_{\rho\rightarrow\infty}\rho e^{-C\rho r_{\rho}^{2}}=e^{-b}\end{eqnarray*}
Using the above equation, \eqref{eq:limiting value m2p_i} and\eqref{eq:relation upper bound on b3 and m2p_i},
it can be shown that \begin{equation}
\lim_{\rho\rightarrow\infty}\lim_{m\rightarrow\infty}\rho E\left(B_{U,i}\right)-m^{2}p_{i}=0\label{eq:upper bound on b3 equals 0}\end{equation}

As a result of \eqref{eq:b3 a bound on the value}, \eqref{eq:lower bound on b3 equals 0}
and \eqref{eq:upper bound on b3 equals 0}, $\lim_{\rho\rightarrow\infty}\lim_{m\rightarrow\infty}b_{3}=0$.

A combination of the analysis in subsections \ref{sub:An-Evaluation-of-b1},
\ref{sub:An-Evaluation-of-b2} and \ref{sub:An-Evaluation-of-b3}
completes this proof.



\begin{thebibliography}{10}
\providecommand{\url}[1]{#1}
\csname url@samestyle\endcsname
\providecommand{\newblock}{\relax}
\providecommand{\bibinfo}[2]{#2}
\providecommand{\BIBentrySTDinterwordspacing}{\spaceskip=0pt\relax}
\providecommand{\BIBentryALTinterwordstretchfactor}{4}
\providecommand{\BIBentryALTinterwordspacing}{\spaceskip=\fontdimen2\font plus
\BIBentryALTinterwordstretchfactor\fontdimen3\font minus
  \fontdimen4\font\relax}
\providecommand{\BIBforeignlanguage}[2]{{%
\expandafter\ifx\csname l@#1\endcsname\relax
\typeout{** WARNING: IEEEtran.bst: No hyphenation pattern has been}%
\typeout{** loaded for the language `#1'. Using the pattern for}%
\typeout{** the default language instead.}%
\else
\language=\csname l@#1\endcsname
\fi
#2}}
\providecommand{\BIBdecl}{\relax}
\BIBdecl

\bibitem{Gupta98Critical}
P.~Gupta and P.~R. Kumar, ``Critical power for asymptotic connectivity in
  wireless networks,'' in \emph{Stochastic Analysis, Control, Optimization and
  Applications}.\hskip 1em plus 0.5em minus 0.4em\relax Boston, MA: Birkhauser,
  1998, pp. 547--566.

\bibitem{Xue04The}
F.~Xue and P.~Kumar, ``The number of neighbors needed for connectivity of
  wireless networks,'' \emph{Wireless Networks}, vol.~10, no.~2, pp. 169--181,
  2004.

\bibitem{Bettstetter04On}
C.~Bettstetter, ``On the connectivity of ad hoc networks,'' \emph{The Computer
  Journal}, vol.~47, no.~4, pp. 432--447, 2004.

\bibitem{Bettstetter02On}
------, ``On the minimum node degree and connectivity of a wireless multihop
  network,'' in \emph{3rd ACM International Symposium on Mobile Ad Hoc
  Networking and Computing}, 2002, pp. 80--91.

\bibitem{Hekmat06Connectivity}
R.~Hekmat and P.~V. Mieghem, ``Connectivity in wireless ad-hoc networks with a
  log-normal radio model,'' \emph{Mobile Networks and Applications}, vol.~11,
  no.~3, pp. 351--360, 2006.

\bibitem{Franceschetti07Random}
M.~Franceschetti and R.~Meester, \emph{Random Networks for
  Communication}.\hskip 1em plus 0.5em minus 0.4em\relax Cambridge University
  Press, 2007.

\bibitem{Meester96Continuum}
R.~Meester and R.~Roy, \emph{Continuum Percolation}, ser. Cambridge Tracts in
  Mathematics.\hskip 1em plus 0.5em minus 0.4em\relax Cambridge University
  Press, 1996.

\bibitem{Arratia90Poisson}
R.~Arratia, L.~Goldstein, and L.~Gordon, ``Poisson approximation and the
  chen-stein method,'' \emph{Statistical Science}, vol.~5, no.~4, pp. 403--434,
  1990.

\bibitem{Barbour03Poisson}
A.~D. Barbour, L.~Holst, and S.~Jason, \emph{Poisson Approximation}.\hskip 1em
  plus 0.5em minus 0.4em\relax Oxford University Press, New York, 2003.

\bibitem{Penrose03Random}
M.~D. Penrose, \emph{Random Geometric Graphs}, ser. Oxford Studies in
  Probability.\hskip 1em plus 0.5em minus 0.4em\relax Oxford University Press,
  USA, 2003.

\bibitem{Stein72A}
C.~M. Stein, ``A bound for the error in the normal approximation to the
  distribution of a sum of dependent random variables,'' in \emph{Proc. Sixth
  Berkeley Symp. Math. Statist. Probab}, vol.~2, 1972, pp. 583--602.

\bibitem{Chen75An}
L.~H.~Y. Chen, ``An approximation theorem for sums of certain randomly selected
  indicators,'' \emph{Z. Wahrsch. Verw. Gebiete}, vol.~33, pp. 69--74, 1975.

\bibitem{Chen75Poisson}
------, ``Poisson approximation for dependent trials,'' \emph{Annals of Applied
  Probability}, vol.~3, pp. 534--545, 1975.

\bibitem{Franceschetti06Critical}
M.~Franceschetti and R.~Meester, ``Critical node lifetimes in random networks
  via the chen-stein method,'' \emph{IEEE Transactions on Information Theory},
  vol.~52, no.~6, pp. 2831--2837, 2006.

\bibitem{Penrose99On}
M.~D. Penrose, ``On k-connectivity for a geometric random graph,'' \emph{Random
  Structures and Algorithms}, vol.~15, no.~2, pp. 145--164, 1999.

\bibitem{Philips89Connectivity}
T.~K. Philips, S.~S. Panwar, and A.~N. Tantawi, ``Connectivity properties of a
  packet radio network model,'' \emph{IEEE Transactions on Information Theory},
  vol.~35, no.~5, pp. 1044--1047, 1989.

\bibitem{Ravelomanana04Extremal}
V.~Ravelomanana, ``Extremal properties of three-dimensional sensor networks
  with applications,'' \emph{IEEE Transactions on Mobile Computing}, vol.~3,
  no.~3, pp. 246--257, 2004.

\bibitem{Balister05Connectivity}
P.~Balister, B.~Bollobs, A.~Sarkar, and M.~Walters, ``Connectivity of random
  k-nearest-neighbour graphs,'' \emph{Advances in Applied Probability},
  vol.~37, no.~1, pp. 1--24, 2005.

\bibitem{Wan04Asymptotic}
P.-J. Wan and C.-W. Yi, ``Asymptotic critical transmission radius and critical
  neighbor number for k-connectivity in wireless ad hoc networks,'' 2004.

\bibitem{Balister09A}
P.~Balister, B.~Bollobs, A.~Sarkar, and M.~Walters, ``A critical constant for
  the k nearest neighbour model,'' \emph{Advances in Applied Probability},
  vol.~41, no.~1, pp. 1--12, 2009.

\bibitem{Tang03An}
A.~Tang, C.~Florens, and S.~H. Low, ``An empirical study on the connectivity of
  ad hoc networks,'' in \emph{IEEE Aerospace Conference}, vol.~3, 2003, pp.
  1333--1338.

\bibitem{Orriss03Probability}
J.~Orriss and S.~K. Barton, ``Probability distributions for the number of radio
  transceivers which can communicate with one another,'' \emph{EEE Transactions
  on Communications}, vol.~51, no.~4, pp. 676--681, 2003.

\bibitem{Miorandi05Coverage}
D.~Miorandi and E.~Altman, ``Coverage and connectivity of ad hoc networks
  presence of channel randomness,'' in \emph{IEEE INFOCOM}, vol.~1, 2005, pp.
  491--502.

\bibitem{Miorandi08The}
D.~Miorandi, ``The impact of channel randomness on coverage and connectivity of
  ad hoc and sensor networks,'' \emph{IEEE Transactions on Wireless
  Communications}, vol.~7, no.~3, pp. 1062--1072, 2008.

\bibitem{Bettstetter04failure}
C.~Bettstetter, ``Failure-resilient ad hoc and sensor networks in a shadow
  fading environment,'' in \emph{IEEE/IFIP International Conference on
  Dependable Systems and Networks}, 2004.

\bibitem{Bettstetter05Connectivity}
C.~Bettstetter and C.~Hartmann, ``Connectivity of wireless multihop networks in
  a shadow fading environment,'' \emph{Wireless Networks}, vol.~11, no.~5, pp.
  571--579, 2005.

\bibitem{Yi06Asymptotic}
C.-W. Yi, P.-J. Wan, X.-Y. Li, and O.~Frieder, ``Asymptotic distribution of the
  number of isolated nodes in wireless ad hoc networks with bernoulli nodes,''
  \emph{IEEE Transactions on Communications}, vol.~54, no.~3, pp. 510--517,
  2006.

\end{thebibliography}
\end{document}